\documentclass{article}
\usepackage{tikz}
\usepackage{subcaption}
\usepackage{amsfonts}
\usepackage{amssymb}
\usepackage{graphicx}
\usepackage{amsmath}
\usepackage{amsthm}
\usepackage{pifont}
\usepackage{xcolor}
\usepackage{enumitem}
\usepackage{mathtools}
\usepackage[normalem]{ulem}

\usepackage[colorlinks=true,linkcolor=blue,citecolor= red,bookmarksopen=true,urlcolor=blue]{hyperref}
\nonstopmode
\newtheorem{theorem}{Theorem}[section]
\newtheorem{proposition}[theorem]{Proposition}
\newtheorem{lemma}[theorem]{Lemma}
\newtheorem*{lemma*}{Lemma}

\newtheorem{corollary}[theorem]{Corollary}

\newtheorem{definition}[theorem]{Definition}
\newtheorem{notation}[theorem]{Notation}
\newtheorem*{standingassumptionI*}{Standing Assumption I}
\newtheorem*{standingassumptionII*}{Standing Assumption II}
\newtheorem{assumption}[theorem]{Assumption}

\theoremstyle{remark}
\newtheorem{example}[theorem]{Example}
\newtheorem{remark}[theorem]{Remark}

\newcommand{\R}{{\mathbb R}}

\newcommand\abs[1]{\left|#1\right|}

\newcommand{\fit}{{\mathcal{F}^{i}_t}}

\newcommand{\fib}{{\mathbb{F}^{i}}}
\newcommand{\hi}{{\mathcal{H}^{i}}}
\newcommand{\ki}{{K_{i}}}

\newcommand{\mi}{{M^{i}}}
\newcommand{\mib}{{M^{i,\infty}}}
\newcommand{\mie}{M^i_{e}}
\newcommand{\mibe}{M^{i,\infty}_{e}}
\newcommand{\pii}{\rho_{i,+}}

\renewcommand{\emph}[1]{\textit{#1}}
\newcommand{\red}[1]{\textcolor{red}{#1}}

\newcommand{\blue}[1]{{\color{blue}#1}}

\newcommand{\mcY}{\mathcal {Y}}
\newcommand{\mcF}{\mathcal {F}}

\newcommand{\mcR}{\rho^{\mathcal Y}_{+}}

\newcommand{\Me}{\mathcal M_e(\mcY)}

\usepackage{graphicx} % Required for inserting images

\begin{document}

\title{Collective completeness and pricing hedging duality}
\date{\today}

\author{
Alessandro Doldi\thanks{Dipartimento di Matematica, Università degli Studi di Milano, Via Saldini 50, 20133 Milano, Italy,
\emph{alessandro.doldi@unimi.it}.}, $\quad$ Marco Frittelli\thanks{Dipartimento di Matematica, Università degli Studi di Milano, Via Saldini 50, 20133 Milano, Italy,
\emph{marco.frittelli@unimi.it}.  }, $\quad$
Marco Maggis\thanks{Dipartimento di Matematica, Università degli Studi di Milano, Via Saldini 50, 20133 Milano, Italy,
\emph{marco.maggis@unimi.it}.  }
}

\maketitle

\begin{abstract}
\noindent This paper builds on \textit{Collective Arbitrage and the Value of Cooperation} by Biagini et al. (2025, forthcoming in \textit{Finance and Stochastics}), which introduced in discrete time the notions of collective arbitrage and super-replication in a multi-agent market framework, where agents may operate in several submarkets and collaborate through risk exchange mechanisms. Expanding on these foundations, we establish a First Fundamental Theorem of Asset Pricing and a collective pricing-hedging duality under different assumptions and with new techniques compared to Biagini et al. (2025). We further introduce the notion of collective replication in order to study collective market completeness and provide a Second Fundamental Theorem of Asset Pricing in this cooperative multi-agent setting.
\end{abstract}



\section{Introduction}
The study of arbitrage and pricing-hedging duality has been a central theme in financial mathematics, evolving from classical one-agent framework to more sophisticated multi-agent settings that incorporate cooperation and collective risk-sharing. This paper extends recent developments in the theory of \textit{collective arbitrage} and \textit{super-replication} by analyzing markets where multiple agents trade within segmented markets and engage in cooperative risk exchanges.
Within financial institutions, market segmentation refers to the partitioning of the overall market into specialized submarkets or trading desks.  Each desk concentrates on a particular financial product or asset class. This division can be motivated by several factors: (i) specialization in specific financial instruments (e.g., equities, fixed income, derivatives, commodities); (ii) a focus on assets from particular geographic regions or countries; and (iii) varying risk tolerances, leading to specialization in assets with different risk profiles.  For a broader discussion of the rationale behind market segmentation, see Carassus (2023) \cite{Carassus23}.

A fundamental contribution in this area was made by Biagini et al. (2025) \cite{BDFFM25}, who introduced the notions of \textit{Collective Arbitrage} and \textit{Collective Super-replication} in a discrete-time financial market. They established a collective version of the Fundamental Theorem of Asset Pricing (FTAP) and demonstrated that cooperation among agents reduces the super-replication cost of contingent claims. By allowing risk-sharing mechanisms, their work showed that agents could leverage financial interdependencies to eliminate inefficiencies and improve their overall market positions. Their results highlighted that even in the absence of individual arbitrage, cooperation could create profitable trading opportunities for the collective, making \textit{No Collective Arbitrage} (NCA) a necessary extension of the classical No Arbitrage (NA) condition.

Frittelli (2025) \cite{F25} extended this framework by generalizing the concept of Collective Arbitrage to a semimartingale market setting, introducing a \textit{No Collective Free Lunch} condition. Unlike the discrete-time framework of \cite{BDFFM25}, this extension considered continuous-time markets with general stochastic price processes. The study established a semimartingale version of the collective FTAP, proving that if a market satisfies the No Collective Free Lunch condition, then there exists a suitable vector of equivalent martingale measures, or more generally of separating measures, preventing arbitrage across agents. Furthermore, the work provided a pricing-hedging duality under cooperation, showing that the value of cooperation persists in more general financial settings.

A complementary perspective on these ideas was presented by Doldi et al. (2024) \cite{DFRG24Frontiers}, who explored \textit{Collective Dynamic Risk Measures}. Their work analyzed how cooperation among agents influences risk measurement over time, focusing on \textit{aggregation properties} and \textit{time consistency}. Unlike the classical risk measure theory, which is often defined for a single agent, their study extended the framework to multiple interacting agents in a segmented market. A key contribution was the introduction of a collective aggregation operator, which accounts for risk transfers among agents, allowing for more efficient capital allocation.

The concept of arbitrage in financial markets originated with de Finetti (1931) \cite{deFinetti31} (we refer to Maggis (2023) \cite{Maggis23} for further discussion on the origins of such a notion in the work of de Finetti) and gained prominence in the mid-1970s through several works (see e.g. Ross (1976) \cite{Ross76}, Harrison and Kreps (1979) \cite{HarrisonKreps79}, Harrison and Pliska (1981) \cite{HarrisonPliska81}, Kreps (1981) \cite{Kreps81}) spanning over 30 years of intensive research, first in discrete, then in continuous time. Encompassing overview are provided by  Delbaen and Schachermayer (2006) \cite{DS2006} and and F\"{o}llmer and Schied (2016)\cite{FollmerSchied2}.\\
%%%%%%%%%%%%%%%%%%%
As already stressed in \cite{BDFFM25}, 
the traditional literature as typified by \cite{Kreps81} defines arbitrage as a single nonzero element in the vector space of marketed bundles $M$, a subset of a locally convex topological vector lattice $\mathcal{X}$, with non-positive price and membership in the positive cone $\mathcal{X}^+$. While this definition is general, it has traditionally been interpreted as single-element arbitrage opportunity within $M$. As previously mentioned, the collective approach introduces a multi-dimensional concept comprising a collection of such arbitrage opportunities, which are entangled by the exchanges among the agents.
\paragraph{Main contributions of the paper}

The primary objective of this paper is to investigate collective arbitrage, as developed in \cite{BDFFM25}, within a discrete-time setting, with a specific emphasis on collective replication and completeness when exchanges are characterized by a finite-dimensional vector space. Furthermore, Section \ref{sec:selfin} introduces the concept of self-financing trading strategies within a collective framework, allowing for the dynamic implementation of inter-agent exchanges.
 Specifically, we establish the following key results:
\begin{enumerate}
    \item The First Fundamental Theorem of Asset Pricing in a collective setting (CFTAP I), presented in Theorem \ref{Th1} and Theorem \ref{FTAP:R}.
    \item A collective pricing-hedging duality, detailed in Theorem \ref{duality:R}.
    \item The Second Collective Fundamental Theorem of Asset Pricing (CFTAP II), formulated in Theorem \ref{completeTH}.
\end{enumerate}

While Theorems \ref{Th1} and \ref{FTAP:R} extend results from \cite{BDFFM25} and serve as fundamental tools for our analysis, the core contributions of this paper lie in points 2. and 3.

Specifically, we develop an appropriate version of the second Collective Fundamental Theorem of Asset Pricing by introducing the concept of collectively replicable claims. We also define No Collective Arbitrage prices for vectors of contingent claims, and in Section \ref{sec3}, we characterize such prices through expectations under equivalent collective martingale measures (see Theorem \ref{BaseFTAPII}) and identify the contingent claims that are collectively replicable (see Proposition \ref{prepcftap2}).
  Notably, our discussion on collectively self-financing strategies, collective replication and completeness, and their equivalent characterizations is entirely novel. Furthermore, we provide theoretical examples and counterexamples in Section \ref{sec:examples}. Theorem \ref{duality:R} extends a result from \cite{BDFFM25} under alternative assumptions and introduces a significantly simpler proof technique. Since the results that closely resemble those in \cite{BDFFM25} are primarily instrumental rather than central to our main contributions, we postpone a detailed discussion of their similarities and differences until we have introduced the necessary notation and stated such  results in a rigorous form: see Remark \ref{remY} and the discussion just after Theorem \ref{THclosure}. 
 %on page \ref{discussiondifferences}.

\subsection{Setting}

Let $\mathcal T=\{0, \dots, T\}$ be the finite set of discrete times and consider a given filtered probability space $(\Omega, \mathcal{F}, \mathbb F ,P)$, with  $\mathbb F =\{\mathcal{F}_t\}_{t \in  \mathcal T}$,  
%\sout{$\mathcal{F}_0=\{\emptyset, \Omega\}$} 
and $\mathcal{F}=\mathcal{F}_T$. Except when explicitly stated otherwise, we also assume that $\mcF_0$ is trivial, namely that $P(A)=0$ or $P(A)=1$ for all $A \in \mcF_0$. 
 We say that a probability measure $Q$ on $(\Omega, \mathcal{F})$ belongs to $\mathcal{P}_e$ if $Q \sim P$, or to $\mathcal{P}_{ac}$ if $Q\ll P$, respectively.
 Unless differently stated, all inequalities between random variables are meant to hold $P$-a.s.. \\

  The (global) securities market comprises a zero-interest rate riskless asset $X^0$ and $J$ risky assets with discounted price processes $X^j=(X^j_t)_{t\in [0,T]}$, $j=1, \dots ,J$, $J\geq 1$. We set $X^0_t=1$ for all $t$.

Consider $N$ agents operating within the market, and assume that each agent $i$, $i=1,\dots,N$, can invest in the riskless asset and in the risky assets $X^j$, $j\in (i)$, where $(i)$ denotes a specified subset of $\{1,\dots,J\}$. We denote the cardinality of $(i)$ by $d_i$, representing the number of risky assets in which agent $i$ may invest, and assume $1 \leq d_i \leq J$.
We assume (without loss of generality) that $\cup_{i=1}^N (i)=\{1,\dots,J\}$, as
we may ignore the assets that can not be used by any agent. We are not excluding that different agents may invest in the same risky assets. Let $X^{(i)}:=(X^j,\, j \in (i))$ be the risky assets available for trading to agent $i$.
We denote by $$\fib=(\fit)_{t\in\mathcal{T}}\subseteq \mathbb F$$ the filtration representing the information available to agent $i$. Thus for all $i$, $\mathcal F^i_0=\mathcal F_0$, the trivial sigma-algebra.\\
We assume that all processes $X^j$, with $ j \in (i)$, are adapted with respect to the  filtration $\fib$.

\begin{notation}
\label{productnotation}
If $Q=(Q^1,\dots,Q^N) \in (\mathcal P_{ac})^{N} $ and $p \in [0,\infty]$ we set $$L^{p }(\Omega, \mathbf{F}_t,Q):=L^{p }(\Omega, \mathcal{F}_t^{1},Q^1) \times \dots \times  L^{p }(\Omega, \mathcal{F}_t^{N},Q^N), \quad t\in\{0,\dots,T\} .$$ 
With a slight abuse of notation, when we consider only one probability measure $Q \in \mathcal P$, we still use the same notation $$L^{p }(\Omega, \mathbf{F}_t,Q):=L^{p }(\Omega, \mathcal{F}_t^{1},Q) \times \dots \times  L^{p }(\Omega, \mathcal{F}_t^{N},Q),\quad t\in\{0,\dots,T\} .$$
\end{notation}

 We say that a set $K \subseteq L^{0 }(\Omega, \mathbf{F}_T,P)$ is closed in $L^{0 }(\Omega, \mathbf{F}_T,P)$ if it is closed for the convergence in probability. This is equivalent to the property that 
${K}\subseteq L^{0}(\Omega, \mathbf{F}_T,P)$ is (sequentially) closed under $P$-a.s. convergence, namely $f_n\to f$ $P$-a.s. and $(f_n)\subset {K}$, implies $f\in \mathcal{K}$.
 We adopt the following partial order among random vectors: for every $f,g\in L^{0}(\Omega, \mathbf{F}_T, P)$ we write $f \leq g$ if and only if $P(f^i\leq g^i)=1$ for every $i=1,\ldots, N$. For any $a\in\R$ we shall always indicate with $\mathbf{a}$ the vector $\mathbf{a}=(a,a,\ldots,a)\in\R^N$. 
 \\

A stochastic process $H=(H_t)_{t\in \mathcal T}$ is called an \label{admisstrat}\emph{admissible trading strategy for the agent} $i$ if it is $d_i$ - dimensional and  predictable with respect to $\fib$.
The space of admissible trading strategies for the agent $i$ is denoted by $\hi$. 
If $H \in \hi$, we set 
\begin{equation*}
(H\cdot X^{(i)})_t:=\sum_{h \in (i) } (H^h\cdot X^h)_t,    
\end{equation*} 
where $(H^h\cdot X^h)_t:=\sum_{s=1}^t H^h_s(X^h_s-X^h_{s-1})$ denotes the stochastic integral of $H^h$ with respect to  the asset $X^h$, $h\in (i)$, and we write

\begin{equation}
\label{def:KiBis}
\ki:=\{(H\cdot X^{(i)})_T \mid H \in {\hi} \} \subseteq L^{0 }(\Omega, \mathcal{F}_T^{i},P).
\end{equation}

Let $\widehat X^i:=(X^0,X^{(i)})$ be the $(d_i+1)$ - dimensional adapted process representing the market available to agent $i$, which also includes the rikless asset $X^0$. For a $1$-dimensional and $\fib$-predictable stochastic process $h$ and for $H \in \mathcal H^i$,  we say that the trading strategy $\widehat H=(h,H)$ is \textit{self financing for agent $i$} (for brevity we write s.f. for $i$) if
$\widehat H_t\widehat X_t^{(i)}=\widehat H_{t+1}\widehat X_t^{(i)}$, $t=1,\dots,T-1$. If $\widehat H$ is s.f. for $i$, we let $V^i_0(\widehat H):=\widehat H_1\widehat X_0^{(i)}$ and $V^i_t(\widehat H):=\widehat H_t\widehat X_t^{(i)}$, $t=1,\dots,T$, and we recall that for any $v_0^i\in \mathbb R$

\begin{equation*}
    \{V^i_T(\widehat H) \mid \widehat H \text{ is s.f. for } i \text{ with } V^i_0(\widehat H)=v_0^i \}=\{v_0^i+k \mid  k \in K_i \},
\end{equation*}
namely the final value of s.f. trading strategies for agent $i$ can be obtained by adding to its initial value $v^i_0$ the stochastic integral of admissible trading strategies $H \in \mathcal H^i$ with respect only to the risky assets $X^{(i)}$.

\noindent The set of martingale measures for  the assets in $(i)$ is defined by
\begin{equation*}
\mi:=\left\{ Q \in \mathcal{P}_{ac} \mid     X^j  \text { is a } (Q, \fib) \text {-martingale for all } j \in (i) \right\},
\end{equation*}
while the set of martingale measures with bounded densities for  the assets in $(i)$ is defined by
\begin{align*}
\mib(P):=\bigg\{ Q \in \mathcal{P}_{ac} \mid \,&\frac {dQ} {dP} \in L^{\infty }(\Omega, \mathcal{F},P)  \\&\text { and } X^j  \text { is a } (Q, \fib) \text {-martingale for all } j \in (i) \bigg\}.
\end{align*}
Moreover, we set $$\mie:=\mathcal P_e \cap \mi \quad \text {and} \quad \mibe(P):=\mathcal P_e \cap \mib(P). $$
The classical No Arbitrage condition for agent $i$ holds if:

\begin{equation*}
\label{eq:NAi}
 \mathbf{NA}_{i}: \text{  } \ki \cap L_{+}^{0}(\Omega, \mathcal{F}^i_T , P)=\{0\}.
 \end{equation*}
 The $\mathbf{NA}_{i}$ condition implies that the set $(\ki - L_{+}^{0}(\Omega, \mathcal{F}^i_T , P))$ is closed in $L^{0}(\Omega, \mathcal{F}^i_T , P)$ (see \cite{DS2006} Theorem 6.9.2), an essential property for the proof of both the FTAP I and the super-hedging duality.
\begin{theorem}[Dalang, Morton and Willinger (1990) \cite{DMW90}]\label{DMW}
In the discrete time setting described above and if $X^{(i)}$ is integrable under $P$ then
\begin{equation*}
 \mathbf{NA}_{i} \;\Longleftrightarrow\;  \mibe(P) \not = \emptyset \;\Longleftrightarrow\; \mie \not = \emptyset .
 %\quad \mathbf{NA} \Longleftrightarrow  M \cap \mathcal{P}_e \not = \emptyset.
\end{equation*}
\end{theorem}

It is well-known that the integrability of $X^{(i)}$
  can be ensured by a change to an equivalent probability measure.

\begin{definition}
For a random variable $f \in L^{0}(\Omega, \mathcal{F}^i_T , P)$, we define the classical super-replication  price for agent $i$ as
\begin{equation*}
\pii(f):=\inf\{m \in \mathbb{R} \mid \exists k \in \ki \text{ s.t. } m+k \geq f \}.
\end{equation*}
\end{definition}
\noindent Under $\mathbf{NA}_i$, the following classical super hedging duality holds true (see \cite{BDFFM25} page 12). 
\begin{equation}\label{superclassic}
\pii(f)=\sup_{Q \in\mibe(P)} E_{Q}[f] \ \text { for all } f \in L^1(\Omega, \mathcal{F}^i_T , P).
\end{equation}

\subsection{Setting for collective arbitrage}

A fundamental principle in the aforementioned reference \cite{BDFFM25} is that $N$ agents, in addition to trading within their individual markets, can enhance their portfolios through collaborative risk exchanges.
Let $\mathcal{Y}_0$ denote the set of all zero-sum risk exchanges, defined as
\begin{equation*} \label{2345}
\mathcal{Y}_0 = \left\{ Y \in L^0(\Omega, \mathbf{F}_T, P) \mid \sum_{i=1}^N Y^i = 0  \right\}.
\end{equation*}

While the sum of the components of any $Y \in \mathcal{Y}_0$ is almost surely zero under $P$, the individual components $Y^i$ are, in general, random variables.  A positive realization of $Y^i$ on a particular event signifies a capital inflow for agent $i$ from the collective, while a negative value indicates a capital outflow.  Thus, $Y \in \mathcal{Y}_0$ characterizes the potential capital reallocations among the agents, subject to the constraint of zero net transfer. An example of permissible exchanges is a subset $\mathcal{Y}$ of $\mathcal{Y}_0$. 

More generally, throughout this paper, we adopt the following
\begin{assumption}\label{assY}
The set $\mcY$ of of allowed exchanges is a subset $$\mcY \subseteq L^0(\Omega, \mathbf{F}_T, P)  \,\text{ with }\, 0 \in \mcY.$$   
\end{assumption}

% More generally, we will only assume that the set of of possible\slash allowed exchanges is a subset $$\mcY \subseteq L^0(\Omega, \mathbf{F}_T, P)  \,\text{ with }\, 0 \in \mcY.$$

We will also often assume that the vector space of al possible deterministic exchanges
\begin{equation*}\label{R0}
\mathbb R_0 ^N:=\left \{   x \in \mathbb R^N \mid \sum_{i=1} ^N x^i =0  \right \}
%=\mathrm{span}  \left \{  (e^i-e^j) \mid  i,j \in \{1, \dots , N \}  \right \},
\end{equation*}
%where $\{e^i\}_i$ is the canonical basis of $\mathbb R ^N$
is contained in $\mcY$.\\

When agent $i$ follows an investment strategy in its own market $(i)$, she will obtain a terminal payoff $k^i \in {K}_i$.
The agents may also enter in the risk exchange corresponding to a vector $Y \in \mathcal Y $. This procedure leads to the terminal time value $k^i + Y^i$ for agent $i$.
In \cite{BDFFM25} a \emph{Collective Arbitrage} consists of vectors $(k^1, \dots , k^N) \in {K}_1 \times \cdots \times {K}_N$ and $Y =(Y^1, \dots, Y^N) \in \mcY$ satisfying
\begin{align*}
 k^i+Y^i\geq 0  \quad P\text{-a.s. } & \quad \forall \, i\in{1,\dots,N} \, \text{ and } 
\\ P(k^n+Y^n>0)>0   &\quad \text{for at least one } n \in {1,\dots,N}.
\end{align*}

 \begin{definition}[No Collective Arbitrage, Def 3.1 \cite{BDFFM25}]
\label{NCA} 
No Collective Arbitrage for $\mathcal{Y}$ ($\mathbf{NCA}(\mathcal{Y})$) holds if

\begin{equation*}
({\sf X}_{i=1} ^{N} K_i+\mathcal Y )\cap  L^{0 }_+(\Omega, \mathbf{F}_T,P)=\{0\}, \label{NCAY} 
 \end{equation*}
 where ${\sf X}_{i=1} ^{N} \ki$ denotes the Cartesian product of the sets $\ki$ defined in \eqref{def:KiBis}.
 \end{definition}

 \begin{proposition}[Proposition 3.2, \cite{BDFFM25}]\label{equivalentCo}
The following conditions are equivalent: 

\begin{align*}
 \mathbf{NCA}(\mathcal{Y}) \\
 K^\mathcal Y \cap  L^{0 }_+(\Omega, \mathbf{F}_T , P)&=\{0\},  \\
 C^\mathcal Y \cap  L^{1 }_+(\Omega, \mathbf{F}_T,P)&=\{0\}, 
 \end{align*}
 where 
 \begin{equation}
 \label{KyandCY}
 K^\mathcal Y:=  {\sf X}_{i=1} ^{N} ( K_i - L^{0 }_+(\Omega, \mathcal{F}^i_T,P) ) + \mathcal Y \quad \text{ and } \quad C^\mathcal Y :=K^\mathcal Y \cap L^{1 }(\Omega, \mathbf{F}_T,P). 
 \end{equation}
 \end{proposition}
 
As proved in Section 3.1 \cite{BDFFM25}, the implications
\begin{equation}\label{Implications}
\mathbf{NA} \Rightarrow \mathbf{NCA(\mathcal{Y})} \Rightarrow \mathbf{NA}_{i} \; \forall\, i\in{1,\dots,N},
\end{equation}
essentially always hold: the former only requires $\mcY\subseteq \mcY_0$, the latter holds under the even weaker assumption that $0\in\mcY$. Moreover, it was shown in the aforementioned reference that
\begin{align}
     \mathbf{NA} & \Longleftrightarrow  \mathbf{NCA}(\mcY) \,\text{ if }\, \mcY=\mcY_0 \text{ and }\mathbb{F}^i=\mathbb{F}^j=\mathbb{F}, \label{12}
   \\  \mathbf{NA}_{i} \text{ } \forall \, i & \Longleftrightarrow   \mathbf{NCA(\mcY)} \,\text{ if }\, \mcY=\mathbb R_0 ^N.\label{NAiNCA}
\end{align}
 However, the notions of $\mathbf{NCA(\mathcal{Y})}$ generate novel concepts for general sets $\mcY$.

The following set of vectors of equivalent martingale measures with bounded densities was adopted in Theorem 3.12 \cite{BDFFM25} to characterize $\mathbf{NCA}(\mcY)$
\begin{align*}
\mathcal M ^{\infty}_e(P)=\bigg\{ Q=(Q^1,\dots,Q^N) \in {\sf X}_{i=1} ^{N} \mibe(P)  \mid &  \sum_{i=1}^N E_{Q^i}[Y^i]\leq 0 \\&\forall Y \in \mcY \cap L^{1 }(\Omega, \mathbf{F}_T,P)  \bigg\}.
\end{align*}

\begin{theorem}[Collective FTAP, Th. 3.12 \cite{BDFFM25}]\label{FTAP3}
 Let $\mcY $ be a convex cone satisfying $R^N_0 \subseteq \mcY \subseteq L^{0 }(\Omega, \mathbf{F}_T,P)$ and such that $K^{\mcY}$ is closed in $L^{0 }(\Omega, \mathbf{F}_T,P)$. 
 If $\mathcal{Y}  \subseteq L^{1 }(\Omega, \mathbf{F}_T , P)$ and $X=(X^1,\dots,X^J)$  is integrable under $P$ then 

\begin{equation*}\label{CFTAP11}
 \mathbf{NCA(\mathcal{Y})} \iff \mathcal M ^{\infty}_e(P) \not = \emptyset.
 \end{equation*}
\end{theorem}

\subsection{A new version of  the collective FTAP I }

 Differently from Theorem \ref{FTAP3}, each equivalent condition in the following Theorem depends on the selection of the underlying probability measure $P$, only through its negligible sets.

\begin{definition}
The set $\Me$ of equivalent collective martingale measures is defined by
\begin{equation}\label{MartingaleMeasures}
\begin{split}
    \Me=\bigg\{ Q=(Q^1,\dots,Q^N) \in {\sf X}_{i=1} ^{N}\mie   \mid & \mcY \subseteq L^{1 }(\Omega, \mathbf{F}_T,Q)  \\&\text {and  } \sum_{i=1}^N E_{Q^i}[Y^i]\leq 0 \text {  }\forall Y \in \mcY  \bigg\}.
\end{split}
\end{equation}
\end{definition}

\begin{theorem}\label{Th1}
   Let $\mcY $ be a convex cone satisfying $R^N_0 \subseteq \mcY \subseteq L^{0 }(\Omega, \mathbf{F}_T,P)$ and such that $K^{\mcY}$ is closed in $L^{0 }(\Omega, \mathbf{F}_T,P)$. If there exists $P' \in \mathcal P_e$ such that $\mcY \subseteq L^{1 }(\Omega, \mathbf{F}_T,P')$ then 
    \begin{equation*}
      \mathbf{NCA(\mathcal{Y})} \iff \Me \not = \emptyset.  
    \end{equation*}
\end{theorem}

\begin{remark}\label{remY}$\,$\\
 \begin{enumerate}
 \item When $\mcY$ is a finitely generated convex cone, the assumption that there exists $P' \in \mathcal P_e$ such that $\mcY \subseteq L^{1 }(\Omega, \mathbf{F}_T,P')$  is clearly redundant.
   \item  We note that differently from $\mathcal M ^{\infty}_e(P)$ appearing in Theorem \ref{FTAP3}, which depends on the underlying historical measure $P$, in the new Theorem \ref{Th1} we exploit the newly defined set of collective martingale measures $\Me$. The latter only depends on the \emph{null sets} of the underlying measure $P$. 
   This also permits greater flexibility by requiring integrability of elements in $\mathcal{Y}$ only up to a change of measure (i.e., the existence of $P' \in \mathcal{P}_e$ such that $\mathcal{Y} \subseteq L^1(\Omega, \mathbf{F}_T, P')$), rather than imposing $\mathcal{Y} \subseteq L^1(\Omega, \mathbf{F}_T, P)$ for the original measure $P$, as was the case in Theorem \ref{FTAP3}.
        \item The Theorem establishes the equivalence between the absence of Collective Arbitrage and the existence of equivalent collective martingale measures, even in the context of distinct filtrations $\mathbb F^i$, one for each agent $i$.
    
        \item The central assumption of the theorem, the closedness of  $K^{\mcY}$ in $L^{0}(\Omega, \mathcal{F}_T, P)$, is addressed in the Section \ref{proofs:R}. It will be demonstrated, Theorem \ref{THclosure}, that $K^{\mcY}$ is indeed closed in $L^{0}(\Omega, \mathcal{F}_T, P)$ under $\mathbf{NCA(\mathcal{Y})}$, provided that $\mcY$ is a finite dimensional vector space.
        
    \end{enumerate}
    
\end{remark}

\begin{proof}

We first prove $\Me \not = \emptyset \Rightarrow \mathbf{NCA(\mathcal{Y})}$.  Take $(Q^1, \dots , Q^ N) \in  \Me$ and let $(k+Y) \in  ({\sf X}_{i=1} ^{N}  K_i   + \mathcal Y ) \cap  L^{0 }_+(\Omega,  \mathbf{F}_T,P)$. Then $k^i+Y^i \geq 0$ and thus $k^i \geq -Y^i  \in L^{1 }(\Omega, \mathcal{F}^i_T , Q^i) $ for all $i$. Hence $(k^i)^- \in L^{1}(\Omega, \mathcal{F}^i_T,Q^i)$ and thus, by Theorem 5.14 \cite{FollmerSchied2},
 $k^i \in L^{1}(\Omega, \mathcal{F}^i_T,Q^i) $ and $E_{Q^i}[k^i]=0$. Therefore $E_{Q^i}[k^i+Y^i] = E_{Q^i}[k^i]+E_{Q^i}[Y^i]=E_{Q^i}[Y^i]$, $\sum_{i=1}^N E_{Q^i}[k^i+Y^i] = \sum_{i=1}^N E_{Q^i}[Y^i]  \leq 0$, as $(Q^1, \dots , Q^ N) \in  \Me$.
From $(k^i+Y^i) \geq 0 $ for all $i$, we also get that $\sum_{i=1}^N E_{Q^i}[k^i+Y^i]  \geq 0$, so that $\sum_{i=1}^N E_{Q^i}[k^i+Y^i] = 0$. From $0 \in  ({\sf X}_{i=1} ^{N}  K_i   + \mathcal Y ) $ then $E_{Q^i}[k^i+Y^i] = 0$ for all $i$ and $k^i+Y^i=0$ for all $i$. Thus $\mathbf{NCA(\mathcal{Y})}$ holds true.

We now prove $\mathbf{NCA(\mathcal{Y})}\Rightarrow \Me \not = \emptyset $. We assume that $\mcY \subseteq L^{1 }(\Omega, \mathbf{F}_T,P')$ for some $P' \in \mathcal P_e$. By setting $\frac{d \widehat P}{dP'}:= \frac{c}{1+\sum_{j,t}| X^j_t |}$,  for some positive normalizing constant $c$, we see that $\widehat P \in \mathcal P_e$, $\frac {d\widehat P} {dP'} \in L^{\infty }(\Omega, \mathcal{F},P')$, $X \subseteq L^{1 }(\Omega, \mathbf{F}_T, \widehat P)$  and $\mcY \subseteq L^{1 }(\Omega, \mathbf{F}_T, \widehat P)$.
 Since $K^{\mcY}$ is closed in $L^{0 }(\Omega, \mathbf{F}_T,P)$,  $C^\mathcal Y := K^\mathcal Y \cap L^{1 }(\Omega, \mathbf{F}_T,\widehat P)$ is closed in $L^{1 }(\Omega, \mathbf{F}_T,\widehat P).$ 
      As $\widehat P \in \mathcal P_e$, we may assume that $\mathbf{NCA(\mathcal{Y})}$ holds under $\widehat P$ and thus from Proposition \ref{equivalentCo} we know that $\mathbf{NCA(\mathcal{Y})}$  is equivalent to 
     $ C^\mathcal Y \cap  L^{1 }_+(\Omega, \mathbf{F}_T,\widehat P)=\{0\}$. Thus
      by the multidimensional version of Kreps-Yan Theorem, Th. 8.3 \cite{BDFFM25}, we deduce the existence of a vector $z=(z^1.\dots,z^N) \in L^{\infty }(\Omega, \mathbf{F}_T,\widehat P)$, $z^i>0$ for all $i$, such that 
    \begin{equation}\label{102bis}
    \sum_{i=1}^N \widehat E[z^if^i] \leq 0 \text { for all } f \in C^\mathcal Y.    
    \end{equation}
    Since $\mcY \subseteq L^{1 }(\Omega, \mathbf{F}_T, \widehat P)$, we may rewrite $C^\mathcal Y$ as
    \begin{equation*}\label{CYbis}
        C^\mathcal Y = K^\mathcal Y \cap L^{1 }(\Omega, \mathbf{F}_T,\widehat P)=  ({\sf X}_{i=1} ^{N} ( K_i - L^{0 }_+(\Omega, \mathcal{F}^i_T, \widehat P) )) \cap L^{1 }(\Omega, \mathbf{F}_T,\widehat P) +\mcY.
    \end{equation*} 
    Condition \eqref{102bis} now implies $\sum_{i=1}^N \widehat E[z^iY^i] \leq 0 \text { for all } Y \in \mathcal Y$ and from $\mathbb R^N_0 \subseteq \mcY$ and from Remark 3.4 \cite{BDFFM25} we get $\widehat E[z^i]=\widehat E[z^j]$ for all $i,j$. Set $\frac {dQ ^i} {d \widehat P}:=\frac {z^i} { \widehat E[z^i]}>0$ for all $i$.  Thus $Q^i \in \mathcal P_e$ for all $i$, $\frac {dQ ^i} {d \widehat P} \in L^{\infty }(\Omega, \mathcal{F},\widehat P)$, $\mcY \subseteq L^{1 }(\Omega, \mathbf{F}_T, Q) $ for $Q=(Q^1,\dots,Q^N)$  and $\sum_{i=1}^N  E_{Q^i}[Y^i] \leq 0 \text { for all } Y \in \mathcal Y$.
    Moreover, \eqref{102bis} implies $\sum_{i=1}^N  
    E_{Q^i}[f^i] \leq 0 \text { for all } f \in ({\sf X}_{i=1} ^{N} ( K_i - L^{0 }_+(\Omega, \mathcal{F}^i_T, \widehat P) )) \cap L^{1 }(\Omega, \mathbf{F}_T,\widehat P)$, so that $ E_{Q^i}[k^i] = 0 \text { for all } k^i \in K_i \cap L^{1 }(\Omega, \mathcal{F}^i_T, \widehat P), $ for all $i$.\\ 
    When $X \subseteq L^{1 }(\Omega, \mathbf{F}_T, \widehat P)$, we recall from  \cite{DS2006}, Section 6.11, or \cite{FollmerSchied2} Theorem 5.14 that $\mib(\widehat P)$ can be written as  
\begin{equation*}\label{MartingaleMeasuresbounded}
\mib(\widehat P)=\left\{ Q \in \mathcal{P}_{ac} \mid \frac {dQ} {d\widehat P} \in L^{\infty }(\Omega, \mathcal{F},\widehat P)  \text { and } E_Q[k]= 0  \text {  } \forall k \in \ki \cap L^{1 }(\Omega, \mathcal{F}_T^{i},\widehat P) \right\}.
\end{equation*}
Since $X \subseteq L^{1 }(\Omega, \mathbf{F}_T, \widehat P)$ we conclude that $Q^i \in \mibe(\widehat P)$. This shows that 
    \begin{align*}
M:=\bigg\{ Q=(Q^1,\dots,Q^N) \in {\sf X}_{i=1} ^{N}\mibe(\widehat P)   \mid \,&\mcY \subseteq L^{1 }(\Omega, \mathbf{F}_T,Q)  \\&\text { and  } \sum_{i=1}^N E_{Q^i}[Y^i]\leq 0 \text {  }\forall Y \in \mcY  \bigg\} \not = \emptyset.
\end{align*}
Since $\widehat P \in \mathcal P_e$ we get $ \mibe(\widehat P) \subseteq \mie $, $ M\subseteq \Me $, thus $\Me \not = \emptyset.$ 
\end{proof}

\begin{remark}\label{remEQ}
    An inspection of the first part of the proof reveals that the implication $\Me \not = \emptyset \Rightarrow \mathbf{NCA(\mathcal{Y})}$ holds under the only Assumption \ref{assY} on the set $\mcY \subseteq L^{0 }(\Omega, \mathbf{F}_T,P)$.
\end{remark}

\subsection{On collectively self financing trading strategies}
\label{sec:selfin}

The aforementioned reference highlights the agents' capacity not only to participate in their respective markets but also to improve their outcomes through collaborative risk exchanges. We now detail how these exchanges can be dynamically implemented and subsequently show that, for collectively self-financing trading strategies, this dynamic aspect does not significantly alter the theory, except when dealing with time consistency (see the example in Section \ref{sec:timeconsist}).

For each $t$, let $\mathcal{Y}(t) \subseteq L^0(\Omega, \mathbf{F}_t, P)$. A random vector $Y_t = (Y_t^1, \dots, Y_t^N) \in \mathcal{Y}(t)$ describes the permissible exchanges among the agents at time $t$.  Recognizing that agents may cooperate at different times, we must adapt the notion of a trading strategy's value process and reformulate the concept of self-financing.

Let $v_0 = (v_0^1, \dots, v_0^N) \in \mathbb{R}^N$ denote the initial wealth of each agent, and let $\mathbf{Y} = (Y_1, \dots, Y_T) \in {\sf X}_{t=1} ^{T} \mcY(t)$  represent a collection of potential exchanges, with $Y_t = (Y_t^1, \dots, Y_t^N)$ executed at each time $t$.

We define a collective trading strategy as a family of processes 
$\mathbf{\widehat{H}} = (\widehat{H}^1, \dots, \widehat{H}^N)$, where $\widehat H^i=(h^i, H^i)$, for each agent $i$.
The process $h^i$ (a one-dimensional, $\mathcal{F}^i$-predictable process) 
represents investment in the riskless asset $X^0$, distinct from $H^i \in \mathcal{H}^i$ which represents investment in the $d_i$ risky assets $X^{(i)}$.
We denote with $\widehat X^i:=(X^0,X^{(i)})$ the $(d_i+1)$ - dimensional process and recall that we are working under the simplifying condition $X^0_t=1$ for every $t$.
The time $t=0$ value for each agent is given by

\begin{equation} \label{V0}
V_0^i(\widehat{H}^i, \mathbf{Y}) = \widehat{H}_1^i \widehat{X}_0^i = v_0^i, \quad i = 1, \dots, N.
\end{equation}
At time $t=1$, the value becomes $\widehat{H}_1^i \widehat{X}_1^i$, and the value after the exchange is

$$V_1^i(\widehat{H}^i, \mathbf{Y}) = \widehat{H}_1^i \widehat{X}_1^i + Y_1^i.$$
The portfolio $\widehat{H}_2^i$ selected by agent $i$ at time $t=1$ may reflect the exchange at time 1. Thus, we require that $\widehat{H}_2^i$ satisfies the condition
$$\widehat{H}_2^i \widehat{X}_1^i = \widehat{H}_1^i \widehat{X}_1^i + Y_1^i,$$
instead of the classical $\widehat{H}_2^i \widehat{X}_1^i = \widehat{H}_1^i \widehat{X}_1^i$.
While $\widehat{H}_2^i$ is chosen by agent $i$ using information available at $t=1$, this condition depends on the collectively chosen exchange $Y_1 = (Y_1^1, \dots, Y_1^N)$ (e.g., through the constraint $\sum_{i=1}^N Y_1^i = 0$, if $Y_1 \in \mcY_0$).  Furthermore, the amount $Y_1^i$ contributes to the financing of the new portfolio $H_2^i$. This leads to the concept of a \textit{collectively self-financing trading strategy}.

\begin{definition}
    Let $v_0=(v_0^1, \dots,v_0^N) \in \mathbb R^N$, $\mathbf Y =(Y_1,\dots,Y_T) \in {\sf X}_{t=1} ^{T} \mcY(t)$ and let $\mathbf {\widehat H} =(\widehat H^1, \dots,\widehat H^N)$ be a collective trading strategy. For each $i=1,\dots,N$ we define $V^i_0(\widehat H^i,\mathbf Y):=\widehat H^i_1\widehat X^i_0$ %v_0^i+Y^i_0$ 
    and 
    \begin{equation*}
V^i_t(\widehat H^i,\mathbf Y):=\widehat H^i_t\widehat X^i_t +Y^i_t,\,\,\, t =1,\dots,T .
\end{equation*} 
Moreover, $\mathbf {\widehat H}$ is called collectively self financing (c.s.f.) if \eqref{V0} holds and
    \begin{equation*}
     \widehat H^i_{t+1}\widehat X^i_t=\widehat H^i_t\widehat X^i_t +Y^i_t, \,\,  t =1,\dots,T-1,\,\, i=1,\dots,N.   
    \end{equation*}
    The vector having $i$ component $V^i_t(\widehat H^i,\mathbf Y)$ is denoted by  $V_t(\mathbf{\widehat H},\mathbf Y)$.
\end{definition}

Mirroring the single-agent case, the following proposition shows that the time-$T$ value $V_T(\mathbf{\widehat{H}}, \mathbf{Y})$ of a collectively self-financing trading strategy $\mathbf{\widehat{H}}$ admits, for each agent $i$, the decomposition \eqref{VV}, 
where $H^i \in \mathcal{H}^i$ is an admissible trading strategy investing solely in the risky assets $X^{(i)}$ (without the self-financing condition), and $Y_{1:T}^i = \sum_{s=1}^T Y_s^i$ represents the cumulative exchanges up to time $T$.

\begin{proposition}
  Let $v_0 \in \mathbb R^N$ and $\mathbf Y =(Y_1,\dots,Y_T) \in {\sf X}_{t=1} ^{T} \mcY(t)$. If  $\mathbf {\widehat H}$ is collectively self financing (c.s.f.) then 
  \begin{equation}\label{VV}
  V^i_t(\widehat H^i,\mathbf Y)=v^i_0+(H^i\cdot X^{(i)})_t + Y^i_{1:t},\,\, t=1,\dots,T,
  \end{equation}
  where $\widehat H^i=(h^i,H^i)$, with $ H^i \in \mathcal H^i$, and $Y^i_{1:t}=\sum_{s=1}^t Y_s^i$. Moreover,
\begin{equation}\label{VV2}
    \{V_T(\mathbf{\widehat H},\mathbf Y) \mid \mathbf {\widehat H} \text{ is c.s.f. and } V_0(\mathbf{\widehat H},\mathbf Y) =v_0 \}=\{v_0+k+Y_{1:T} \mid  k \in {\sf X}_{i=1} ^{N} K_i, \mathbf{Y} \in {\sf X}_{t=1} ^{T} \mcY(t)  \}.
\end{equation}
\end{proposition}
\begin{proof}
    Using the definition of $V^i_t(\widehat H^i,\mathbf Y)$ compute 
    \begin{align*}
      V^i_t(\widehat H^i,\mathbf Y)-V^i_{t-1}(\widehat H^i,\mathbf Y)&= \widehat H^i_t\widehat X^i_t +Y^i_t-(\widehat H^i_{t-1}\widehat X^i_{t-1} +Y^i_{t-1})\\
      &=\widehat H^i_t\widehat X^i_t +\widehat H^i_t\widehat X^i_{t-1} -\widehat H^i_t\widehat X^i_{t-1} +Y^i_t-\widehat H^i_{t-1}\widehat X^i_{t-1} -Y^i_{t-1}\\
      &=\widehat H^i_t(\widehat X^i_{t}-\widehat X^i_{t-1})+Y^i_t=H^i_t( X^{(i)}_{t}- X^{(i)}_{t-1})+Y^i_t,
    \end{align*}
    where the second-to-last equality follows from the c.s.f. condition: $\widehat H^i_{t}\widehat X^i_{t-1}=\widehat H^i_{t-1}\widehat X^i_{t-1} +Y^i_{t-1}$, and we applied $X^0_t=1$, for all $t$, in the latter equality.  Equation \eqref{VV} then follows from
    \begin{equation}\label{difference}
     V^i_t(\widehat H^i,\mathbf Y)-V^i_0(\widehat H^i,\mathbf Y)=\sum_{s=1}^t   (V^i_s(\widehat H^i,\mathbf Y)-V^i_{s-1}(\widehat H^i,\mathbf Y) ).
    \end{equation}
 As a consequence of \eqref{difference}, to prove \eqref{VV2} we only need to show the inclusion $ (\supseteq)$ in \eqref{VV2}. Let $v_0 \in \mathbb R^N$, $k \in {\sf X}_{i=1} ^{N} K_i, \mathbf{Y} \in {\sf X}_{t=1} ^{T} \mcY(t)$ be given with $k^i=(H^i\cdot X^{(i)})_T \in K_i$, $ H^i \in \mathcal H^i$ for all $i$. We need to construct a c.s.f. $\mathbf {\widehat H}$ in the form,  $\widehat H^i=(h^i,H^i)$, with $h^i$ a $1$-dimensional and $\fib$-predictable process, and $V^i_0(\widehat H^i,\mathbf Y):=\widehat H^i_1\widehat X^i_0=v_0^i$, for all $i$. By \eqref{VV} this would then imply $V^i_T(\mathbf {\widehat H},\mathbf Y)=v^i_0+(H^i\cdot X^{(i)})_T + Y^i_{1:T}=v_0^i+k^i+Y^i_{1:T}$, that is the thesis. To construct such $\mathbf {\widehat H}$, from the condition $v_0^i=\widehat H^i_1\widehat X^i_0=h_1^iX_0^0+H^i_1X^{(i)}_0$, one can obtain $h_1^i$, which is $\mathcal{F}^i_0$-measurable. Let $t \geq 1$ and suppose now that $  (h^i_{1}, \ldots, h^i_{t}) $ is $\mathcal{F}^i$-predictable and $\widehat H^i=(h^i,H^i)$ is c.s.f. up to time $ (t-1) $. The condition $\widehat H^i_t\widehat X^i_t +Y^i_t=\widehat H^i_{t+1}\widehat X^i_t=h_{t+1}^iX_t^0+H^i_{t+1}X^{(i)}_t$ allows to compute the $\mathcal{F}^i_t$-measurable random variable $h_{t+1}^i$, showing that $  (h^i_{1}, \ldots, h^i_{t+1}) $ is $\mathcal{F}^i$-predictable and $\widehat H^i=(h^i,H^i)$ is c.s.f. up to time $ t $.
 This shows how to construct such  $\mathbf {\widehat H}$ up to the final time.
    
\end{proof}

For each $t$, recall that $\mathcal{Y}(t) \subseteq L^0(\Omega, \mathbf{F}_t, P)$. For $\mathbf{Y}=(Y_1,\dots,Y_T) \in {\sf X}_{t=1} ^{T} \mcY(t) $ we set $Y^i_{1:t}:=\sum_{s=1}^t Y_s^i$, $i=1, \dots,N$ and  $\mcY_{1:T}:=\{(Y^1_{1:T},\dots, Y^N_{1:T}) \mid \mathbf{Y} \in {\sf X}_{t=1} ^{T} \mcY(t) \}$. 
Thus, in vector notation, the final wealth $V_T(\mathbf{\widehat{H}}, \mathbf{Y})$ of all agents using collectively self-financing strategies with zero initial cost can be expressed as

\begin{equation}\label{eqVT}
 V_T(\mathbf {\widehat H},\mathbf Y)=k+Y \in {\sf X}_{i=1} ^{N} K_i +\mcY_{1:T}.   
\end{equation}

Several specific cases warrant consideration.

\begin{enumerate}

\item Restrictions on the timing of exchanges can be modeled by imposing measurability constraints on the random vector $Y \in \mathcal{Y}_{1:T}$ in Equation \eqref{eqVT}. Specifically, if $\mathbb{F}^i=\mathbb{F}^j=\mathbb{F}$ and  the set of admissible exchanges at the final time $T$ is chosen such that $\mathcal{Y}(T) = \mathcal{Y}_0$, then $\mathcal{Y}_{1:T} = \mathcal{Y}_0$, and we recover the case where $\mathbf{NCA}(\mathcal{Y}_{1:T})$ if and only if $\mathbf{NA}$ (see \eqref{12}). Alternatively, if exchanges are prohibited after time $\widehat{t} \in \{1, \dots, T-1\}$, then $\mathcal{Y}_{1:T} \subseteq L^0(\Omega, \mathbf{F}_{\widehat{t}}, P)$, which corresponds to the setting discussed in Section 6.3 of \cite{BDFFM25}.

  \item Similarly, additional constraints on the measurability of each $Y_t = (Y_t^1, \dots, Y_t^N) \in L^0(\Omega, \mathbf{F}_t, P)$ can be imposed. For example, one might require $(Y_t^1, \dots, Y_t^N) \in L^0(\Omega, \mathbf{F}_{t-1}, P)$ for each $t$. In the extreme case where exchanges are permitted only at time $t=1$, this corresponds to allowing only deterministic exchanges, that is $\mathcal{Y}_{1:T}=\mathbb R^N_0$, and thus $\mathbf{NCA}(\mathcal{Y}_{1:T})$ if and only if $\mathbf{NA}_i$ for all $i$ (see \eqref{NAiNCA}).

    \item To establish relevant examples of admissible exchanges for which the convex cone $K^{\mathcal{Y}}$ is closed in probability, we will require specific geometric properties of the set $\mathcal{Y}(t)$. For instance, we might require that $\mathcal{Y}(t)$ be a \textit{finite dimensional  vector space}. Assuming that  $\cap_{i=1}^N \mathcal F ^i_t \not= \emptyset$ and that $\{A_1, \dots, A_M\}$, with $A_n \in \cap_{i=1}^N \mathcal F ^i_t $
    %\mathcal{F}_t$ 
    for each $n$, is a finite partition of $\Omega$, 
   % $\mathbb{F}^i=\mathbb{F}^j=\mathbb{F}$, 
    a relevant example is
    \begin{equation*}
    \mathcal{Y}(t) := \left\{ Y_t  \in L^0(\Omega, \mathbf{F}_t, P) \mid Y_t^i = \sum_{n=1}^M y_n^i 1_{A_n}, \sum_{i=1}^N y_n^i = 0, y_n \in \mathbb{R}^N \right\}.
    \end{equation*}
\end{enumerate}

Based on the preceding discussion and the representation of the final wealth $V_T(\mathbf{\widehat{H}}, \mathbf{Y})$ of agents using collectively self-financing strategies with zero initial cost as in \eqref{eqVT}, we will henceforth assume that the final wealth of each agent $i$ can be expressed in the form
$$k^i + Y^i,$$
for some $k^i \in K_i$ and some vector $Y = (Y^1, \dots, Y^N) \in \mathcal{Y} \subseteq L^0(\Omega, \mathbf{F}_T, P)$.

%%%%%%%%%%%%%%%%

\section{Pricing-hedging duality, closure of $K^{\mcY}$ and completeness for finite dimensional exchanges}\label{sec2}

In this section we propose a friendly proof of the Collective Fundamental Theorem of Asset Pricing (CFTAP I) and the pricing-hedging duality, in a setup where exchanges among agents are allowed only on a finite dimensional vector space. 
Let $(\Omega, \mathcal{F}_T, P)$ be a (general) probability space and recall the meaning of $L^{0}(\Omega, \mathbf{F}_T, P)$ as introduced in  Notation \ref{productnotation}.

\begin{theorem}[CFTAP I]\label{FTAP:R}
    Let $\mcY \subseteq L^{0}(\Omega, \mathbf{F}_T,P)$ be a finite dimensional  vector space containing  $\mathbb R^N_0 $. Then 
    \begin{equation*}
      \mathbf{NCA(\mathcal{Y})} \iff \Me \not = \emptyset.  
    \end{equation*}
where $\mathcal{M}_e(\mcY)$ is the set of collective equivalent martingale measures defined in \eqref{MartingaleMeasures}.
\end{theorem}

The validity of this theorem would be a direct consequence of Theorem \ref{Th1} and Item 1 in Remark \ref{remY}, contingent upon establishing the closedness of $K^{\mcY}$  in $L^{0}(\Omega, \mathbf{F}_T, P)$. This closedness will be demonstrated via a recursive procedure on the the number of vectors generating  $\mcY$ in the subsequent section. As a direct consequence of this theorem, we present in Corollary \ref{cor1} a simplified version of the CFTAP I applicable to vector spaces $\mcY$, which are not necessarily finitely generated.

For each $Y$ in $L^0(\Omega, \mathbf{F}_T,P)$,  recall that $\mathrm{span}(Y,\mathbb R^N_0 ) \subseteq L^{0 }(\Omega, \mathbf{F}_T,P)$ is the (finite dimensional) linear space generated by the single random vector $Y$ and the vector space $\mathbb R^N_0 $.  Observe that the set $\mathcal{M}_e(\mathrm{span}(Y,\mathbb R^N_0 ))$ is given by
\begin{equation*}
\left\{ Q=(Q^1,\dots,Q^N) \in {\sf X}_{i=1} ^{N}\mie   \mid Y \subseteq L^{1 }(\Omega, \mathbf{F}_T,Q)  \text { and  } \sum_{i=1}^N E_{Q^i}[Y^i]\leq 0   \right\}.
\end{equation*}
Having this in mind, we can state the following result.
\begin{corollary}\label{cor1}
    Let $\mcY \subseteq L^{0 }(\Omega, \mathbf{F}_T,P)$ be a vector space containing  $\mathbb R^N_0 $. Then 
    \begin{equation*}
      \mathbf{NCA(\mathcal{Y})} \iff  \mathcal M_e(\mathrm{span}(Y,\mathbb R^N_0 )) \neq \emptyset \quad \forall Y \in \mcY.  
    \end{equation*}
\end{corollary}

\begin{proof}
From Theorem \ref{FTAP:R} we know that $$ \mathbf{NCA}(\mathrm{span}(Y,\mathbb R^N_0 )) \iff \mathcal M_e(\mathrm{span}(Y,\mathbb R^N_0 )) \neq \emptyset. $$
   From the definition of $\mathbf{NCA}$, one can check that 
   \begin{equation*}
    \mathbf{NCA}(\mathrm{span}(Y,\mathbb R^N_0 )) \quad \forall Y \in \mcY \iff \mathbf{NCA}(\mathcal{Y} +\mathbb R^N_0  )  \iff \mathbf{NCA}(\mathcal{Y}),
   \end{equation*}
   where the last equivalence is trivial as, by assumption, $(\mathcal{Y} +\mathbb R^N_0  )=\mathcal{Y}$.
\end{proof}
\begin{remark}
Analogously, Corollary \ref{cor1} remains valid, with appropriate modifications, when a finite collection of vectors from $\mcY$ replaces the single vector $Y \in \mcY$.
\end{remark}

In order to state the collective version of the pricing-hedging duality, we recall from \cite{BDFFM25} the concept of collective super/sub replication price.

\begin{definition}[Definition 4.1 and 4.16 \cite{BDFFM25}]
    For $\mcY\subseteq L^{0}(\Omega, \mathbf{F}_T,P)$ and $(f^1,\dots,f^N)=f\in L^{0}(\Omega, \mathbf{F}_T, P)$ we define the collective superhedging and subhedging price as
\begin{align}
    \label{super:rho} \rho^{\mcY}_+(f) & =\inf\left\{\sum_{i=1}^N m_i \mid m\in\R^N \text{ and }  m+k+Y \geq f \;\text{for }  k\in {\sf X}_{i=1} ^{N}  K_i,\, Y\in \mcY \right\},
    \\ \notag \rho^{\mcY}_-(f) & =\sup\left\{ \sum_{i=1}^N m_i\mid m\in\R^N \text{ and }  m+k-Y \leq f \;\text{for }  k\in {\sf X}_{i=1} ^{N}  K_i,\, Y\in \mcY \right\}.
\end{align}
We say that $\rho^{\mcY}_+(f)$ (similarly for $\rho^{\mcY}_-(f)$) defined in \eqref{super:rho} is attained if there exists $ m\in\R^N, \, k\in {\sf X}_{i=1} ^{N}  K_i,\, Y\in \mcY $  such that $m+k+Y \geq f $ and $\rho^{\mcY}_+(f)=\sum_{i=1}^N m_i $ .
\end{definition}

We defer the reader to \cite{BDFFM25} Section 4 for the interpretation of $\rho^{\mcY}_+(f)$ and the difference with the classical superhedging price, there denoted by $\rho^N_+(f)$, of the $N $ claims $(f^1,\dots,f^N)=f$.
We also  note that when $\mcY$ is a vector space 
\begin{equation}
\label{rem:rhopm}
    \rho^{\mcY}_-(f)=-\rho^{\mcY}_+(-f)\text{ for every }f\in L^{0}(\Omega, \mathbf{F}_T,P)
\end{equation}

To handle the super-hedging duality of potentially unbounded or non-integrable contingent claim vectors $f \in L^0(\Omega, \mathbf{F}, P)$ we introduce the following convex set of vectors of martingale measures. For any arbitrary $\varphi\in L^{0}(\Omega, \mathcal{F},P)$ let
 \begin{equation}\label{Mfi}\mathcal{M}^{\varphi}_e(\mcY):=\{Q\in \mathcal{M}_e(\mcY)\mid E_{Q_i}[|\varphi^i|]<\infty\;\forall\,i=1,\ldots,N\} \subseteq \mathcal{M}_e(\mcY) .
 \end{equation}
 
\begin{remark}\label{remdegenerate}
    We observe that for $\mcY=\mathrm{span}(Y_1,\dots, Y_R)$ and  $\varphi^i=\max \{|Y^i_1|,\ldots, |Y^i_R|\}$ we have $\mathcal{M}^{\varphi}_e(\mcY)=\Me $. 
\end{remark}
 
 \begin{assumption}
 \label{ass2} $\,$

     \begin{enumerate}
         \item $\mcY $ is a finite dimensional vector space satisfying $R^N_0 \subseteq \mcY \subseteq L^{0 }(\Omega, \mathbf{F}_T,P)$, i.e. $\mcY=\mathrm{span}(Y_m,m=0,\dots,R)$, with $Y_0=R^N_0$ and $Y_m \in L^{0 }(\Omega, \mathbf{F}_T,P)$, $m=1,\dots,R$.
         \item \textbf{NCA}$(\mcY)$ holds true
     \end{enumerate}
 \end{assumption}

The Assumption \ref{ass2} in particular implies that there exists $P' \in \mathcal P_e$ such that $\mcY \subseteq L^{1 }(\Omega, \mathbf{F}_T,P')$ 
\begin{theorem}[Pricing hedging duality]\label{duality:R} Let Assumption \ref{ass2} holds true. Then 
for any $f\in L^{0}(\Omega, \mathbf{F}_T,P)$ and for any $\varphi\in L^{0}(\Omega, \mathbf{F}_T,P)$ such that  
$\max \{|f^i|,|Y^i_1|,\ldots, |Y^i_R|\}\leq \varphi^i$ we have $\mathcal{M}^{\varphi}_e(\mcY) \not= \emptyset$ and
\begin{equation}\label{pricing:hedging}
  \rho^{\mcY}_+(f)=\sup \left \{\sum_{i=1}^N E_{Q^i}[f^i]\mid Q\in \mathcal{M}^{\varphi}_e(\mcY) \right \}>-\infty. 
\end{equation}
  Moreover, when finite, $ \rho^{\mcY}_+(f) $ defined in \eqref{super:rho} is attained.
%\end{itemize}
  
\end{theorem}

The proofs of Theorem \ref{FTAP:R} and Theorem \ref{duality:R} are postponed to Section \ref{proofs:R}. As a byproduct (see Section \ref{proofs:R}, Item 1) we will also obtain the following

\begin{theorem}[Closure of $K^{\mcY}$]\label{THclosure}
    Under Assumption \ref{ass2}, $K^{\mcY} \text{ is closed in } L^{0}(\Omega, \mathbf{F}_T,P)$.
    \end{theorem}

\label{discussiondifferences} A comparison with the findings of \cite{BDFFM25} is now presented.
The results in this paper hold for possibly distinct filtrations, one for each agent. Although this formulation is already present in \cite{BDFFM25}, many of the proofs therein require coinciding filtrations—a condition we do not impose here. In particular, with respect to the closure of $K^{\mcY}$, this represents a significant departure from the earlier approach.
Aside from the case of coinciding filtrations, the requirements and the approach here are generally alternative and complementary in spirit to those in \cite{BDFFM25}. Specifically, we keep defining trading strategies and their final payoffs in full generality, but we establish our results by imposing finite dimensionality on the set $\mcY$ of permitted exchanges—a condition not considered in \cite{BDFFM25}—while avoiding the scalability condition stipulated in Requirement 2 of \cite{BDFFM25} Theorem 6.11. The finite dimensionality allows us to adopt a recursion method which was not feasible in the typical setup of \cite{BDFFM25}, and was inspired by the proof of Theorem 5.1 in  Bouchard and Nutz (2015) \cite{BouchardNutz15}.  Moreover, in the dual representation \eqref{pricing:hedging}, we emphasize that our framework permits a broader range of choices for the sets of measures on the right-hand side, since the vector $\varphi$ can be selected with considerable freedom, thereby enabling the use of a rather wide range of  subsets of collective martingale measures.

In Section \ref{sec3} we analyze when a segmented market with $N$ agent is collectively complete and prove the following collective version of the second FTAP.
\noindent Recall that a contingent claim $f^i \in L^{0}(\Omega, \mathcal{F}_T^i, P)$ is replicable for agent $i$ if $f^i \in \R+  K_i $. We shall see that the following notion of $\mcY$-collectively replicable vector of contingent claims turns out to be crucial.

\begin{definition}
 The vector of contingent claims $ f=( f^1,\dots,  f^N) \in L^{0}(\Omega, \mathbf{F}_T, P)$    is $\mcY$-collectively replicable if $ f\in \R^N+ {\sf X}_{i=1} ^{N} K_i+\mathcal Y  $.
\end{definition}

\begin{theorem}[CFTAP II]\label{completeTH} 
Under Assumption \ref{ass2}
the market is collectively complete if and only if  $\Me$  is a singleton, namely $$\text{every } f\in L^{0}(\Omega, \mathbf{F}_T,P) \text{ is }  \mcY\text{-collectively replicable}  \Longleftrightarrow \Me \text{ is a singleton}.$$
\end{theorem}

\subsection{Auxiliary results}
We first observe the following algebraic property of superhedging functional: let $\mcY\subseteq L^{0}(\Omega, \mathbf{F}_T, P)$ be a vector space containing $\R^N_0$ and $f\in L^{0}(\Omega, \mathbf{F}_T,P)$ then 
\begin{equation}\label{rhoKY}
     \rho^{\mcY}_+(f)  =\inf\left\{\sum_{i=1}^N m_i \mid m\in\R^N \text{ and }   f-m \in {K}^{\mcY} \right\},
\end{equation}
where  $K^{\mcY}$ is defined in \eqref{KyandCY}.

\begin{lemma}\label{lemmarho}
   Suppose that $\mcY\subseteq L^{0}(\Omega, \mathbf{F}_T, P)$ is a vector space containing $\R^N_0$ and $\mathbf{NCA(\mcY)}$ holds true. Then
   
\begin{enumerate}
\item $\mcR (0)=0$.
\item $\mcR$ is cash additive: $\mcR(f+c)=\mcR(f) +\sum_{i=1} ^N c^i$, for all  $c \in \mathbb R^N$ and $f \in L^{0}(\Omega, \mathbf{F}_T, P)$.
\item $\mcR$ is monotone increasing with respect to the partial order of $L^{0 }(\Omega, \mathbf{F}_T,P) $ and positively homogeneous.
\item $\mcR(f)<+\infty$ for any bounded from above $f \in  L^{0 }(\Omega, \mathbf{F}_T,P)  $
\item Suppose  additionally that ${K}^{\mcY}$
is closed in $L^{0}(\Omega, \mathbf{F}_T, P)$ and $f \in L^{0}(\Omega, \mathbf{F}_T, P)$. Then $\rho^{\mcY}_+(f)>-\infty$ and $\rho^{\mcY}_-(f)<+\infty$. If $\rho^{\mcY}_+(f)<+\infty$ (resp. $\rho^{\mcY}_-(f)>-\infty$) then $\rho^{\mcY}_+(f)$ (resp. $\rho^{\mcY}_-(f)$) is attained. More specifically, $$ f-\frac{\mcR (f)}{N}\mathbf 1 \in K^{\mcY}.$$ 
\item  $\rho^{\mcY}_-(f)\leq \rho^{\mcY}_+(f)$ for every $f\in L^{0}(\Omega, \mathbf{F}_T,P)$.

\end{enumerate}

\end{lemma}
\begin{proof}
    Items 1-3 were proved in Lemma 4.4 \cite{BDFFM25} and Item 4 is a simple consequence of Items 1, 2 and 3.
    Recall that $\mathbf{a}=(a,a,\ldots,a)\in\R^N$, for $a \in \mathbb R$.
    We come to Item 5.
    First we recall from Proposition 4.9 \cite{BDFFM25} that  $\rho^{\mcY}_+(f)=N\cdot \pi^{\mcY}_+(f)$ with
\begin{equation}\label{piY}
 \pi^{\mcY}_+(f)  =\inf\left\{a\in\R\mid f\leq \mathbf{a}+k+Y \;\text{for }  k\in {\sf X}_{i=1} ^{N}  K_i,\, Y\in \mcY \right\}.
 \end{equation}
 Assume that $\rho^{\mcY}_+(f)=-\infty$ so that $\pi^{\mcY}_+(f)=-\infty$. Then for every $n\in \mathbb{N}$ there exists $k_n,Y_n$ such that $f^i\leq -n+k^i_n+Y^i_n$ for every $i=1,\ldots,N$. This implies $(f^i+n)\wedge 1\leq f^i+n\leq k^i_n+Y^i_n$ for every $i=1,\ldots,N$. We have $(f+\mathbf{n})\wedge \mathbf{1}\in {K}^{\mcY}$ for every $n\in\mathbb{N}$. Moreover $(f+\mathbf{n})\wedge \mathbf{1}\to \mathbf{1}$ $P$-a.s. and therefore $\mathbf{1}\in K^{\mcY}$, which contradicts $\mathbf{NCA(\mcY)}$.
\\ Now we proceed proving the attainment when $\rho^{\mcY}_+(f)$ is finite. By definition, there exists
    a sequence $(m_n)_n$ in  $\mathbb R^N$ such that $\widehat m_n=:\sum_{i=1}^N m^i_n \downarrow \mcR (f)$ and $f-m_n \in K^\mathcal Y $. Subsequently, writing $f-m_n=f-\frac{\widehat m_n}{N}\mathbf 1 +\frac{\widehat m_n}{N} \mathbf 1-m_n $, we see that also $f-\frac{\widehat m_n}{N}\mathbf 1 \in  K^\mathcal Y$. Indeed, $(m_n -\frac{\widehat m_n}{N}\mathbf 1) \in \R^N_0 \subseteq K^\mathcal Y$. 
    The thesis follows from the closure of $K^\mathcal Y$ and $f-\frac{\widehat m_n}{N}\mathbf 1 \rightarrow f-\frac{\mcR (f)}{N}\mathbf 1. $

 We now come to Item 6. We observe that if $\rho^{\mcY}_+(f)=+\infty$ or if $\rho^{\mcY}_-(f)=-\infty$ there is nothing to prove. We can then suppose that $\rho^{\mcY}_+(f)<+\infty$ and $\rho^{\mcY}_-(f)>-\infty$ as well, i.e. the infimum in defining $\rho^{\mcY}_+(f)$ and the supremum in defining $\rho^{\mcY}_-(f)$ are both taken over nonempty sets.
 One can introduce, similarly to $\pi^\mcY_+(f)$, $$\pi^{\mcY}_-(f)  =\sup\left\{a\in\R\mid f\geq \mathbf{a}+k-Y \;\text{for }  k\in {\sf X}_{i=1} ^{N}  K_i,\, Y\in \mcY \right\}.$$ Similarly to \eqref{rem:rhopm} one can verify that $\pi^\mcY_-(f)=-\pi^\mcY_+(-f)$ and hence $\rho^\mcY_-(f)=N\pi^\mcY_-(f)$. We then just need to verify $\pi^\mcY_-(f)\leq \pi^\mcY_+(f)$ when the infimum in defining $\pi^{\mcY}_+(f)$ and the supremum in defining $\pi^{\mcY}_-(f)$ are both taken over nonempty sets.
 Observe that if $\mathbf{x}+k+Y\geq f$ and $\mathbf{z}+\tilde{k}-\tilde{Y}\leq f$ for $x,z\in\R$ $k,\tilde{k}\in {\sf X}_{i=1} ^{N}  K_i$ and $Y,\tilde{Y}\in \mcY$ 
 we have $(\mathbf{x}-\mathbf{z})+(k-\tilde{k})+(Y+\tilde{Y})\geq f+(-f)=0$ which implies $(k^i-\tilde{k}^i)+(Y^i+\tilde{Y}^i)\geq z-x$ for every $i$. Since $\mathbf{NCA}(\mcY)$ holds we must then have $z\leq x$. Taking suprema over $z$ and infima over $x$ the inequality $\pi^\mcY_-(f)\leq \pi^\mcY_+(f)$ follows.   
\end{proof}

In Proposition 4.12 \cite{BDFFM25} it was shown that if $f \in L^{1}(\Omega, \mathbf{F}_T, P)$,   $\mcY\subseteq L^{1}(\Omega, \mathbf{F}_T, P)$ is a convex cone containing $\R^N_0$, if $\mathbf{NCA(\mcY)}$ holds, if ${K}^{\mcY}$ is closed under $P$-a.s. convergence and if $\rho^{\mcY}_+(f)<+\infty$, then $\rho^{\mcY}_+(f)$ is attained. Item 5 in Lemma \ref{lemmarho} shows, using a simpler direct proof, that attainment holds also if we remove the integrability requirements on $\mcY$ and $f$.

\begin{lemma}\label{interval} Let $\widetilde{\mcY} \subseteq L^{0}(\Omega, \mathbf{F}_T,P)$, be a vector space containing $\R^N_0$ and $\mcY = \mathrm{span}\{\widetilde{\mcY},f\}$ with $f \in L^{0}(\Omega, \mathbf{F}_T,P)$ non $\widetilde{\mcY}-$ collectively replicable. If $\mathbf{NCA(\mcY)}$ holds true and if ${K}^{\widetilde{\mcY}}$ is closed under $P$-a.s. convergence then $\rho^{\widetilde\mcY}_-(f) < 0 <  \rho^{\widetilde\mcY}_+(f)$.
\end{lemma}

\begin{proof}
  Assume by contradiction that  $\rho_+:=\rho^{\widetilde{\mcY}}_+(f)\leq 0$. As  $\mathbf{NCA(\mcY)} \Rightarrow \mathbf{NCA(\widetilde{\mcY})}$, by Lemma \ref{lemmarho} Item 5 we know $\rho^{\widetilde{\mcY}}_+(f)>-\infty$ and is attained. Then we can find $\rho_+^i \in \mathbb R$, $k_+\in {\sf X}_{i=1} ^{N}  K_i$, $Y_+\in \widetilde\mcY$ such that $f^i\leq \rho_+^i + k_+^i + Y_+^i$ for every $i=1,\ldots,N$, and $\sum_i \rho_+^i=\rho_+\leq 0$. This implies
\[-f^i+ \rho_+^i + Y^i_+ + k^i_+ \geq 0 \quad \forall \,i=1,\ldots,N \]
and at least one inequality being strict, otherwise $f$ would be $\widetilde{\mcY}-$ collectively replicable. Consider now $x^i= \rho_+^i - \frac{\rho_+}{N}$, then $x^i\geq \rho_+^i$ and 
\[-f^i+ x^i + Y^i_+ + k^i_+ \geq  0 \quad \forall \,i=1,\ldots,N, \]
with one inequality being strict. Moreover $\sum_i x^i=0$ and therefore $-f+x+Y_+\in \mcY = \mathrm{span}\{\widetilde{\mcY},f\}$ which guarantees that $\mathbf{NCA(\mathcal{Y}})$ is violated. To conclude, we observe that since both $ {\sf X}_{i=1} ^{N}  K_i $ and $\mcY$ are vector spaces, $f$ is $\mcY$-collectively replicable if and only if so is $-f$. In particular, the previous argument shows that $\rho^{\widetilde{\mcY}}_+(-f)>0$ and we can use \eqref{rem:rhopm} to conclude $\rho^{\widetilde{\mcY}}_-(f)=-\rho^{\widetilde{\mcY}}_+(-f)<0$ as desired.
\end{proof}

\begin{lemma} \label{closure+Y} Let $\widetilde{\mcY} \subseteq L^{0}(\Omega, \mathbf{F}_T,P)$, be a vector space with $\R^N_0\subseteq \widetilde{\mcY}$ and $\mcY = \mathrm{span}(\widetilde{\mcY},Y)$ for $Y\in L^{0}(\Omega, \mathbf{F}_T, P)$. If $\mathbf{NCA(\mcY)}$ holds true and ${K}^{\widetilde{\mcY}}$ is closed under $P$-a.s. convergence then ${K}^{\mcY}$  is closed under $P$-a.s. convergence.
\end{lemma}

\begin{proof}
 If $Y=\widetilde{k}+\widetilde{Y}$ for  $\widetilde{k}\in {\sf X}_{i=1} ^{N} K_i$ and $\widetilde{Y}\in \widetilde{\mcY}$ then we prove ${K}^{\mcY}={K}^{\widetilde{\mcY}}$, which is closed by assumption. Indeed ${K}^{\mcY}\supseteq {K}^{\widetilde{\mcY}}$. If $f\in {K}^{\mcY}$ then there exists $k\in {\sf X}_{i=1} ^{N}  K_i$ and $\widehat{Y}\in \mcY$ such that $f\leq k+\widehat{Y}$. We can write $\widehat{Y}=W+\alpha Y$ with $W\in\widetilde{\mcY}$ and $\alpha\in\R$. Therefore $f\leq k+W+\alpha Y=k+W+\alpha (\widetilde{k}+\widetilde{Y})=(k+\alpha\widetilde{k})+(W+\alpha \widetilde{Y})$, hence $f\in {K}^{\widetilde{\mcY}}$.
\\ Now we suppose that $Y\notin {\sf X}_{i=1} ^{N} K_i+\widetilde{\mcY}$ and we prove that ${K}^{\mcY}$ is closed. Let $f_n \in {K}^{\mcY}$ converging to $f$ $P$-a.s.. Then 
\begin{equation}\label{fkx}
f_n\leq  k_n+Y_n+\beta_n Y,    
\end{equation}
for $k_n \in {\sf X}_{i=1} ^{N} K_i$ and $Y_n\in \widetilde{Y}$ and $\beta_n \in \R$. 

If there exists a subsequence which we indicate again with $(\beta_n)$ which converges to a finite $\beta$ then $f_n-\beta_n Y$ converges $P$-a.s. to $f-\beta Y$. Clearly $f_n-\beta_n Y\leq  k_n+Y_n$ i.e. $f_n-\beta_n Y\in {K}^{\widetilde{\mcY}}$ so that $f-\beta Y\in {K}^{\widetilde{\mcY}}$ i.e. $f-\beta Y\leq  k+\widetilde{Y}$ and therefore $f\in {K}^{\mcY}$.

Suppose now there exist no subsequence converging to a finite $\beta$, but that we can find a subsequence $(\beta_n)$ such that $\lim \beta_n =+\infty$. In this case we get $\frac{f_n}{\beta_n} \rightarrow 0$. From \eqref{fkx} we get 
\[\frac{f_n}{\beta_n}- Y\leq  \frac{1}{\beta_n}(k_n+Y_n)\]
Thus  $(\frac{f_n}{\beta_n}- Y) \in {K}^{\widetilde{\mcY}}$ and $\frac{f_n}{\beta_n}- Y  \rightarrow -Y$, therefore $-Y\in {K}^{\widetilde{\mcY}}$. In particular $-Y=k-l+\widetilde{Y}$, for some $\widetilde{Y} \in \widetilde\mcY$, $k \in {\sf X}_{i=1} ^{N}  K_i $, $l \in L^{0 }_+(\Omega,  \mathbf{F}_T,P)$ and  we get  $k+Y+\widetilde{Y}=l\geq 0$,  with $(Y+\widetilde{Y}) \in \mcY$. Then $\mathbf{NCA(\mcY)}$ implies $k+Y+\widetilde{Y}=0$. As $f_n\leq  k_n+Y_n+\beta_n Y$ we deduce $f_n\leq  k_n+Y_n+\beta_n (-k-\widetilde{Y})$ i.e. $f_n\in {K}^{\widetilde{\mcY}}$ and by closure $f\in {K}^{\widetilde{\mcY}}\subseteq {K}^{\mcY}$. Alternatively,  we can find  a subsequence $(\beta_n)$ such that $\lim \beta_n =-\infty$. In this case we get again $\frac{f_n}{\beta_n} \rightarrow 0$, and from \eqref{fkx} we get 
\[\frac{f_n}{-\beta_n}+ Y\leq  \frac{1}{-\beta_n}(k_n+Y_n)\]
This means $\frac{f_n}{-\beta_n}+ Y\in \mathcal{K}^{\widetilde{\mcY}}$ and therefore $Y\in \mathcal{K}^{\widetilde{\mcY}}$. In particular $Y=k-l+\widetilde{Y}$, for some $\widetilde{Y} \in \widetilde\mcY$, $k \in {\sf X}_{i=1} ^{N}  K_i $, $l \in L^{0 }_+(\Omega,  \mathbf{F}_T,P)$ and  we get $k-Y+\widetilde{Y}=l\geq 0$,  with $(-Y+\widetilde{Y}) \in \mcY$. Then $\mathbf{NCA(\mcY)}$ implies $k-Y+\widetilde{Y}=0$, which contradicts $Y\notin {\sf X}_{i=1} ^{N} K_i+\widetilde{\mcY}$. 
\end{proof}

\subsection{Proofs of Theorems \ref{FTAP:R}, \ref{duality:R} and \ref{THclosure} }\label{proofs:R}

For any finite integer $r$ we set $$\mcY^r:=\mathrm{span}(Y_m, \, m=0,\dots,r ), $$ for $Y_0:=\mathbb R^N_0 $  and for a finite number of vectors $Y_m \in L^{0}(\Omega, \mathbf{F}_T, P), $ $m=1,\dots,r $, and we define 
 \begin{equation*}\label{CY}
        {K}^r =  \Big({\sf X}_{i=1} ^{N} \big( K_i - L^{0 }_+(\Omega, \mathcal{F}^i_T, P) \big)\Big)  +\mcY^r.  
    \end{equation*}
We notice that any finite dimensional vector space $\mcY \subseteq L^{0}(\Omega, \mathbf{F}_T, P)$ containing $\mathbb R^N_0 $ can be written as $\mcY=\mcY^R$, for some fixed integer $R$.

We prove at the same time Theorems \ref{FTAP:R}, \ref{duality:R} and \ref{THclosure} by a recursive argument on the number of generating exchanges among agents. 
Thus we fix $\mcY=\mcY^R$, for some fixed integer $R$, and we proceed recursively on $r=0,1,\ldots, R$  to prove the following

\medskip 

\textbf{Claim:}
\textit{If $\mathbf{NCA}(\mathcal{Y}^R)$ holds true then for every $r=0,1,\ldots, R$
\begin{enumerate}
\item ${K}^r$ is closed under $P$-a.s. convergence;
\item $\mathcal{M}_e(\mcY^r) \not= \emptyset$  ;
\item For any $f\in L^{0}(\Omega, \mathbf{F}_T,P)$, $\mathcal{M}^{\varphi}_e(\mcY^r)\not=\emptyset$ and
$$\rho^{\mcY^r}_+(f)=\sup\left \{\sum_{i=1}^N E_{Q^i}[f^i]\mid Q\in\mathcal{M}^{\varphi}_e(\mcY^r) \right \},$$  
for any $\varphi\in L^{0}(\Omega, \mathbf{F}_T,P)$ such that  
$\max\{|f^i|,|Y^i_1|,\ldots, |Y^i_R|\}\leq \varphi^i$.    
\end{enumerate}}

Notice that we already proved in Lemma \ref{lemmarho} Item 5 that when ${K}^r$ is closed under $P$-a.s. convergence then $\rho^{\mcY^r}_+(f)$ defined in \eqref{super:rho} is, when finite, attained.
Moreover, we already know from Remark \ref{remEQ} that $\Me \not = \emptyset \Rightarrow \mathbf{NCA(\mathcal{Y})}$.
Thus Theorem \ref{FTAP:R}, Theorem \ref{duality:R} and Theorem \ref{THclosure} immediately follow from the claim.
Observe also that the assumption $\mathbf{NCA}(\mathcal{Y}^R)$ implies $\mathbf{NCA}(\mathcal{Y}^r)$ for any $r=0,\dots,R$.  \\
\paragraph{Step $r=0$.}
As we are fixing $r=0$, we have $\mcY^r=\R^N_0$. We first demonstrate that 
 \[{K}^r={\sf X}_{i=1} ^{N} ( K_i - L^{0 }_+(\Omega, \mathcal{F}^i_T, P)) +\R^N_0\]
is closed under $P$-a.s. convergence. Let $\{f_n\}\subset K^r$ such that $f_n\to f$ $P$-a.s.. Then
\[f_n\leq x_n+k_n \text{ with } \sum_{i=1}^N x^i_n =0,\]
and $k_n \in {\sf X}_{i=1} ^{N} K_i$, Observe that $\inf_n x^i_n>-\infty$ for any $i=1,\ldots,N$. Otherwise, consider for some $i$ a subsequence, again denoted by ($x^i_n)$, $x^i_n\to -\infty$. For $l_n=x^i_n + k^i_n- f_n^i\geq 0$ and $h_n=\frac{1}{1+|x^i_n|}$ we have 
\[h^n k^i_n-h_n l_n=h_n f^i_n - h_nx^i_n \to_n 1, \]
so that $1\in ( K_i - L^{0 }_+(\Omega, \mathcal{F}^i_T, P)$, which contradicts $\mathbf{NA}_i$. Since $\sum_{i=1}^N x^i_n=0$ and $\inf_n x^i_n>-\infty$ for any $i=1,\ldots,N$, we also have that $\limsup_n x^i_n<+\infty$ for any $i=1,\ldots,N$. This means that we can find a subsequence which we denote again with $(x_n)$ such that $\lim_n x_n=x\in\R^N_0$. Then $f_n-x_n\in {\sf X}_{i=1} ^{N} ( K_i - L^{0 }_+(\Omega, \mathcal{F}^i_T,P) )$ and $f_n-x_n$ converges $P$-a.s. to $f-x$ which therefore belongs to ${\sf X}_{i=1} ^{N} ( K_i - L^{0 }_+(\Omega, \mathcal{F}^i_T,P) )$, i.e. $f\in K^r$.  

Fix any $f \in L^{0}(\Omega, \mathbf{F}_T,P)$ and let  $\varphi \in L^{0}(\Omega, \mathbf{F}_T,P)$ such that
$|f|\leq \varphi$.
We first find a change of measure $\widehat P \sim P$ so that $\varphi \in L^{1}(\Omega, \mathbf{F}_T,\widehat P)$ and all processes $X^1,\dots,X^J$ are integrable under $ \widehat P$.
Thus $f \in L^{1}(\Omega, \mathbf{F}_T,\widehat P)$. The assumption of the claim implies that  $\mathbf{NCA(\mathcal{Y}}^r)$ holds true under $ P$ and thus also under  $\widehat P$.
As $r=0$, from  $\mcY^r=\R^N_0$ and from \eqref{NAiNCA}, we have that $\mathbf{NCA(\mathcal{Y}}^r)$ under $\widehat P$ is equivalent to $\mathbf{NA}_i$ under $\widehat P$, for every  $i=1,\dots,N $,  which holds, by Theorem \ref{DMW}, if and only if 
$M_e^{i,\infty}(\widehat P)\neq \emptyset $ or $M_e^{i}\neq \emptyset $,  for every $i=1,\dots,N $. 
As $r=0$, from  $\mcY^r=\R^N_0$ and from the definition of $\mathcal{M}_e(\mcY^r)$ we have $\mathcal{M}_e(\mcY^r)={\sf X}_{i=1} ^{N}M_e^i$.
We thus conclude that  $\mathcal{M}_e(\mcY^r)\neq \emptyset$.

Take $Q \in {\sf X}_{i=1} ^{N}M_e^{i,\infty}(\widehat P)$, so that $\frac{dQ^i}{d\widehat P} \in L^{\infty}$ for each $i$.
Therefore $E_{Q^i}[\varphi^i]=E_{\widehat P}[\frac{dQ^i}{d\widehat P} \varphi^i ]<\infty$, $Q \in \mathcal{M}_e^{\varphi}(\mcY^r)$ and  $\mathcal{M}_e^{\varphi}(\mcY^r) \not= \emptyset$.

As a consequence $f \in L^{1}(\Omega, \mathbf{F}_T, \mathbf Q)$.
Theorem 5.14 \cite{FollmerSchied2} implies that $E_{Q^i}[k^i]=0$ for all $Q \in \mathcal{M}_e^{\varphi}(\mcY^r)$ and for any $k^i \in K_i $ satisfying $f^i \leq m+k^i$, for some $m \in \mathbb R$. By the classical superhedging duality \eqref{superclassic}, we therefore have 
\begin{align*}
   \rho_{i,+}(f^i)&= \inf\{m\in\R\mid  m+k^i \geq f^i, k^i \in K_i \}  =  \sup\{E_{Q^i}[f^i]\mid Q^i \in M^{i,\infty}_e(\widehat P\} 
    \\ & =  \sup\{E_{Q^i}[f^i]\mid Q^i \in M^{i,\infty}_e(\widehat P) \text{ s.t. } \varphi^i\in L^1(\Omega,\mathcal{F}^i_T,Q^i)\} \\
    &\leq  \sup\{E_{Q^i}[f^i]\mid Q^i \in M^{i}_e \text{ s.t. } \varphi^i\in L^1(\Omega,\mathcal{F}^i_T,Q^i)\} \\
    &\leq    \inf\{m\in\R\mid  m+k^i \geq f^i, k^i \in K_i \}
\end{align*} 
and we conclude, for $M^{i,\varphi^i}_e=\{Q^i \in M^{i}_e \text{ s.t. } \varphi^i\in L^1(\Omega,\mathcal{F}^i_T,Q^i) \}$,
\begin{equation*}
    \rho_{i,+}(f^i):=\inf\{m\in\R\mid  m+k^i \geq f^i, k^i \in K_i \}=\sup\{E_{Q^i}[f^i]\mid Q^i \in M^{i,\varphi^i}_e \}.
\end{equation*}
Observe that since $\mathbf{NA}_i$ holds, $\rho_{i,+}(f^i)>-\infty$ for every $i$. 
When $\mcY^r=\R^N_0$ ($r=0$) $\rho^{\mcY^r}_+(f)$ coincides with the classical super-replication of $N$ claims, denoted by $\rho_+^N(f)$ in  \cite{BDFFM25}. Thus, as explained in \cite{BDFFM25} Proposition 4.3, we have $\rho^{\mcY^r}_+(f)=\sum_{i=1}^N \rho_{i,+}(f^i)$, which is always well defined, so that
\[\rho^{\mcY^r}_+(f)=\sum_{i=1}^N\sup \left \{E_{Q^i}[f^i]\mid Q^i \in M^{i,\varphi^i}_e \right \}=\sup  \left \{\sum_{i=1}^N E_{Q^i}[f^i]\mid  Q\in\mathcal{M}^{\varphi}_e(\mcY^r) \right \}.\]

%%%%%%%%%%%%%%%%%%%%%%%%%%%%%%

\paragraph{Iteration step: $r-1 \Rightarrow r$.} Under $\mathbf{NCA}(\mathcal{Y}^R)$ we assume that the three Items 1. 2. 3. of the claim hold  for $r-1$, and prove they hold  for $r$.

\medskip

We will apply Lemmas \ref{interval} and \ref{closure+Y} to the space $\mcY^r=\mathrm{span}(\mcY^{r-1},Y_r)$. By assumption $\mathbf{NCA}(\mathcal{Y}^R)$ holds and  ${K}^{r-1}$ is closed under $P$-a.s. convergence. As $\mathbf{NCA}(\mathcal{Y}^R)$  implies $\mathbf{NCA}(\mathcal{Y}^{r})$, from Lemma \ref{closure+Y} we deduce that ${K}^r$ is closed under $P$-a.s. convergence, which proves Item 1 of the claim.

\medskip

As proved in Lemma \ref{closure+Y}, if $Y_r\in {\sf X}_{i=1} ^{N} K_i+{\mcY}^{r-1}$ 
then we have ${K}^{r}={K}^{r-1}$. In such a case it  follows by the definitions that $\mathcal{M}^{\varphi}_e(\mcY^r)=\mathcal{M}^{\varphi}_e(\mcY^{r-1}) \not= \emptyset$, so that also $ \mathcal{M}_e(\mcY^r)\not=\emptyset $.  Then from \eqref{rhoKY} we get $\rho^{\mcY^{r-1}}_+(f)=\rho^{\mcY^r}_+(f)$ for any $f\in L^{0}(\Omega, \mathbf{F}_T,P)$ and therefore the duality follows directly from the assumption. We just proved that Items 2 and 3 hold true in this case.\\

Now we suppose that $Y_r\notin {\sf X}_{i=1} ^{N} K_i+{\mcY}^{r-1}$ and
observe that from the induction assumptions, the property $\rho^{\mcY^{r-1}}_-(Y_r)=-\rho^{\mcY^{r-1}}_+(-Y_r)$, and  Lemma \ref{interval} we have
\begin{equation}
\label{inf<sup}
\begin{split}
  a&:=\inf_{Q\in \mathcal{M}_e^{\varphi}(\mcY^{r-1})}\sum_{i=1}^N E_{Q_i}[Y^i_r]=\rho^{\mcY^{r-1}}_-(Y_r)\\
  &<0<\rho^{\mcY^{r-1}}_+(Y_r)=\sup_{Q\in \mathcal{M}_e^{\varphi}(\mcY^{r-1})}\sum_{i=1}^N E_{Q_i}[Y^i_r]:=b.   
\end{split}
\end{equation}
Thus we can find $Q\in \mathcal{M}_e^{\varphi}(\mcY^{r-1})$ such that $Y^i_r \in L^{1}(\Omega, \mathcal{F}^i_T,Q^i)$ and $\sum_{i=1}^N E_{Q_i}[Y^i_r]=0$, which proves $\mathcal{M}_e^{\varphi}(\mcY^{r})\neq \emptyset$. This  implies Item 2 of the claim: $\mathcal{M}_e(\mcY^r)\neq \emptyset$.

\medskip

Finally for any $f\in L^{0}(\Omega, \mathbf{F}_T,P)$ we show
$\rho^{\mcY^r}_+(f)=\sup\{\sum_{i=1}^N E_{Q^i}[f^i]\mid Q\in\mathcal{M}^{\varphi}_e(\mcY^r)\}$,  
for any $\varphi\in L^{0}(\Omega, \mathbf{F}_T,P)$ such that  
$\max\{|f^i|,|Y_1^i|,\ldots, |Y_R^i|\}\leq \varphi^i$.
We have already shown that $\mathcal{M}_e^{\varphi}(\mcY^{r})\neq \emptyset$ and clearly 
\begin{equation}
\label{ineq1rhoy}
    \rho^{\mcY^r}_+(f)\geq\sup\{\sum_{i=1}^N E_{Q_i}[f^i]\mid Q\in\mathcal{M}_e^{\varphi}(\mcY^r)\}.
\end{equation}

 \textbf{Case 1: $\rho^{\mcY^r}_+(f)<+\infty$. } In this case we may for simplicity assume w.l.o.g. (by Lemma \ref{lemmarho} Item 2) that $\rho^{\mcY^r}_+(f)=0$. We claim that 
\begin{equation}\label{claim:0} 0\in E:=\overline{\left\{ \left(\sum_{i=1}^N E_{Q^i}[Y^i_r],\sum_{i=1}^N E_{Q^i}[f^i] \right)\mid Q\in \mathcal{M}_e^{\varphi}(\mcY^{r-1})\right\}}\subset \R^2\end{equation}
By contradiction $0\notin E$. By separation arguments then there exists $(h_1,h_2)\in \R^2$ such that 
\[\sup_{Q\in \mathcal{M}_e^{\varphi}(\mcY^{r-1})}\left\{h_1 \sum_{i=1}^N E_{Q^i}[Y^i_r]+h_2\sum_{i=1}^N E_{Q^i}[f^i] \right\}<0.\] 
In particular, without loss of generality, we may assume that  $\sqrt{h_1^2+h_2^2}=\frac12$ and this implies $\abs{h_1Y^i_r+h_2f^i}\leq \varphi^i$. Moreover by assumption
\begin{eqnarray*}
    \sup_{Q\in \mathcal{M}_e^{\varphi}(\mcY^{r-1})}\left\{h_1 \sum_{i=1}^N E_{Q^i}[Y^i_r]+h_2\sum_{i=1}^N E_{Q^i}[f^i] \right\} & = & \sup_{Q\in \mathcal{M}_e^{\varphi}(\mcY^{r-1})}\left\{\sum_{i=1}^N E_{Q_i}[h_1Y^i_r+h_2f^i]\right\}
    \\ & = & \rho_+^{\mcY^{r-1}}(h_1 Y_r+h_2 f)
\end{eqnarray*}
and by definition and using $Y_r \in \mcY^r$
\[\rho_+^{\mcY^{r-1}}(h_1 Y_r+h_2 f)\geq \rho_+^{\mcY^r}(h_1 Y_r+h_2 f)=\rho_+^{\mcY^r}(h_2 f)\]
i.e. $\rho_+^{\mcY^r}(h_2 f)<0$ which necessarily implies $h_2\neq 0$, by Item 1 Lemma \ref{lemmarho}. If $h_2>0$ then by positive homogeneity (Lemma \ref{lemmarho} Item 3) $\rho_+^{\mcY^r}(f)<0$ which contradicts $\rho_+^{\mcY^r}(f)=0$.
Then necessarily $h_2<0$. As we proved $\mathcal{M}_e^{\varphi}(\mcY^{r})\neq \emptyset$, we take  $\widehat{Q}\in \mathcal{M}_e^{\varphi}(\mcY^r)\subseteq \mathcal{M}_e^{\varphi}(\mcY^{r-1})$ and compute
\[0>\sup_{Q\in \mathcal{M}_e^{\varphi}(\mcY^{r-1})}\left\{\sum_{i=1}^N E_{Q^i}[h_1Y^i_r+h_2f^i]\right\}\geq \sum_{i=1}^N E_{\widehat{Q}^i}[h_1Y^i_r+h_2f^i]= h_2\sum_{i=1}^N E_{\widehat{Q}^i}[f^i].\]
This implies $$\sum_{i=1}^N E_{\widehat{Q}^i}[f^i]>0=\rho^{\mcY}_+(f)\geq\sup_{Q\in\mathcal{M}_e^{\varphi}(\mcY^{r})}\sum_{i=1}^N E_{Q^i}[f^i], $$
a contradiction.
\\The claim in \eqref{claim:0} leads to a sequence of $Q^n\in \mathcal{M}_e^{\varphi}(\mcY^{r-1})$ such that 
$$\sum_{i=1}^N E_{Q_n^i}[Y^i_r]\to 0 \text{ and } \sum_{i=1}^N E_{Q_n^i}[f^i] \to 0.$$ 
We now show that from the conditions in \eqref{inf<sup}
we can construct $Q_n^+,Q_n^-\in \mathcal{M}_e^{\varphi}(\mcY^{r-1})$ and $0\leq \lambda^+_n,\lambda^-_n,\lambda_n$ with $\lambda^+_n+\lambda^-_n+\lambda_n=1$, such that 
$$\widehat{Q}_n:=\lambda^+_n Q_n^+ + \lambda^-_n Q_n^- + \lambda_n Q_n \in \mathcal{M}_e^{\varphi}(\mcY^{r})$$ 
and $\lambda^+_n\to 0, \lambda^-_n\to 0$. In fact for any $\varepsilon>0$ there exists an $n(\varepsilon)$ such that $q_n:=\sum_{i=1}^N E_{Q_n^i}[Y^i_r]\in (-\varepsilon,\varepsilon)$ for any $n>n(\varepsilon)$. We choose in particular $\varepsilon>0$ such that $(-\varepsilon,\varepsilon)\subset (a,b)$ for $a,b$ given in \eqref{inf<sup}. 
For any $n>n(\varepsilon)$ such that $q_n\in (-\varepsilon,0]$ we set $\lambda_n^-=0$, $Q_n^-$ arbitrary, $Q_n^+:=\bar{Q}\in \mathcal{M}_e^{\varphi}(\mcY^{r-1})$ where $\bar{Q}$ satisfies $\bar{q}:=\sum_{i=1}^N E_{\bar{Q}^i}[Y^i_r]\in (\varepsilon, b)$. Therefore in this case we can also find $\lambda^+_n, \lambda_n\in [0,1]$ with $\lambda^+_n+\lambda_n=1$ and $\lambda_n\bar{q}+\lambda_n^+q_n=0$. $\lambda_n\bar{q}+(1-\lambda_n)q_n=0$ implies $\lambda_n= \frac{-q_n}{\bar{q}-q_n}$ and since $q_n\to 0$, then $\lambda_n\to 0$.    
On the other hand, for any $n>n(\varepsilon)$ such that $q_n\in (0,\varepsilon)$, we set $\lambda_n^+=0$, $Q_n^+$ arbitrary, $Q_n^-:=\hat{Q}\in \mathcal{M}_e^{\varphi}(\mcY^{r-1})$ where $\hat{Q}$ satisfies $\hat{q}:=\sum_{i=1}^N E_{\bar{Q}^i}[Y^i_r]\in (a,-\varepsilon)$. Therefore in this case we can also find $\lambda^-_n, \lambda_n\in [0,1]$ with $\lambda^-_n+\lambda_n=1$ and $\lambda_n\hat{q}+\lambda_n^- q_n=0$. $\lambda_n\hat{q}+(1-\lambda_n)q_n=0$ implies $\lambda_n= \frac{-q_n}{\bar{q}-q_n}$ and since $q_n\to 0$ then $\lambda_n\to 0$.     Therefore
$$\sum_{i=1}^N E_{\widehat{Q}_n^i}[Y^i_r] = 0 \text{ and } \sum_{i=1}^N E_{\widehat{Q}_n^i}[f_i] \to 0,$$
which together with \eqref{ineq1rhoy} yields $\rho^{\mcY^r}_+(f)=\sup\{\sum_{i=1}^N E_{Q_i}[f_i]\mid Q\in\mathcal{M}_e^{\varphi}(\mcY^{r})\}$.\\

 \textbf{Case 2: $\rho^{\mcY^r}_+(f)=+\infty$. }
For any $n\in\mathbb{N}$ denote with $f\wedge \mathbf{n}$ the vector of components $f^i\wedge n$ for $i=1,\ldots,n$. By Item 4 Lemma \ref{lemmarho}, $\rho^{\mcY^r}_+(f\wedge \mathbf{n})<+\infty $ and so from the previous step we have $\rho^{\mcY^r}_+(f\wedge \mathbf{n})=\sup\{\sum_{i=1}^N E_{Q^i}[f^i\wedge n]\mid Q\in \mathcal{M}_e^{\varphi}(\mcY^{r})\}$. Since $\rho^{\mcY^r}_+(f\wedge \mathbf{n})$ is increasing the $\lim_n \rho^{\mcY^r}_+(f\wedge \mathbf{n}):=L$ exists. 
By monotonicity of $E_{Q^i}[f^i\wedge n]$, we deduce

\[\sup \left \{\sum_{i=1}^N E_{Q^i}[f^i\wedge n]\mid Q\in\mathcal{M}_e^{\varphi}(\mcY^{r}) \right \}\uparrow_n \sup \left \{\sum_{i=1}^N E_{Q^i}[f^i]\mid Q\in\mathcal{M}_e^{\varphi}(\mcY^{r}) \right \}=L.\]
Note that the terms in LSH equal $\rho^{\mcY^r}_+(f\wedge \mathbf{n})$, so that we conclude: $\lim_n\rho^{\mcY^r}_+(f\wedge \mathbf{n})=L$.
If $L$ were finite 
%$L=\lim_n \rho^{\mcY^r}_+(f\wedge \mathbf{n})$, 
then for $\mathbf{p}^n=(\frac{\rho^{\mcY^r}_+(f\wedge \mathbf{n})}{N},\ldots, \frac{\rho^{\mcY^r}_+(f\wedge \mathbf{n})}{N})$,
the sequence $f\wedge \mathbf{n}-\mathbf{p}^n$ would belong ${K}^r$, by Item 5 Lemma \ref{lemmarho}. As $K^{\mcY}$ is closed, this would imply  $f-\mathbf{\frac{L}{N}}\in {K}^r$ i.e. $L\geq \rho^{\mcY}_+(f)=+\infty$, which is not possible. Hence $L=\sup \left \{\sum_{i=1}^N E_{Q^i}[f^i]\mid Q\in\mathcal{M}_e^{\varphi}(\mcY^{r}) \right \}= +\infty$ and this conclude the proof.

\section{Collectively complete markets: CFTAP II}\label{sec3}

In this section, we provide a precise definition of $\mathbf{NCA}(\mathcal{Y})$-prices for vectors of contingent claims and characterize them, in Theorem \ref{BaseFTAPII}, as expectations under equivalent collective martingale measures. Subsequently, we characterize contingent claims that are collectively replicable, demonstrating that, similarly to the classical setting, $f$ is collectively replicable if and only if $\rho_-^{\mathcal{Y}}(f) = \rho_+^{\mathcal{Y}}(f)$, as shown in Proposition \ref{prepcftap2}. Relying on the dual representation \eqref{pricing:hedging}, we then characterize collective completeness in Proposition \ref{propclosed}.

To begin with, recall that the (global) securities market comprises a zero-interest rate riskless asset $X^0$ and $J$ risky assets with discounted price processes $X^j=(X^j_t)_{t\in [0,T]}$, $j=1, \dots ,J$, $J\geq 1$. 

\begin{definition}
A vector of prices $\Pi_ f \in \mathbb{R}^N$ is called a  \textbf{NCA}$(\mcY)$-price for the contingent claims $f=(f^1,\dots,f^N) \in L^{0 }(\Omega, \mathbf{F}_T,P)$ if there exist processes $X^{J+1},\dots, X^{J+N}$ such that:
\begin{enumerate}
\item $X^{J+i}$ is adapted  to the  filtration $\fib$, $i=1,\dots, N$.
    \item $\Pi^i_ f = X^{J+i}_0$, $i=1,\dots, N$.
    \item $X^{J+i}_T =    f^i$, $i=1,\dots, N$.
    \item For all $i$, \textbf{NCA}$(\mcY)$ holds in the extended market $(X^j)_{j\in (i)\cup\{J+i\}}$.
\end{enumerate} 
Let $\Pi( f)$ be the set of \textbf{NCA}$(\mcY)$-price for the contingent claims $ f=( f^1,\dots,  f^N).$ 
 \end{definition}

\begin{theorem}
\label{BaseFTAPII}
 Let Assumption \ref{ass2} hold. 
Then, for any $ f\in L^{0}(\Omega, \mathbf{F}_T,P)$, and any fixed $\varphi\in L^{0}(\Omega, \mathbf{F}_T,P)$ such that  
$\max \{|f^i|,|Y^i_0|,\ldots, |Y^i_R|\}\leq \varphi^i$, it holds that
\[
\Pi( f) = \left\{ (E_{Q^1}[   f^1],\dots, E_{Q^N}[   f^N]) \mid Q \in \mathcal{M}^\varphi_e(\mcY)\right\},
\]
where $\mathcal{M}^\varphi_e(\mcY)$ is given in \eqref{Mfi}.
\end{theorem}
\begin{proof} 
Let $\varphi\in L^{0}(\Omega, \mathbf{F}_T,P)$ be given as in the statement.
Up to a change of measure to an equivalent one with bounded density, we can assume that $\varphi\in L^{1}(\Omega, \mathbf{F}_T,P') $, $\mcY \subseteq L^{1 }(\Omega, \mathbf{F}_T,P')$ and $X=(X^1,\dots,X^J)$ is integrable under $P'$.
Suppose that $\Pi_ f \in \Pi( f)$, then \textbf{NCA}$(\mcY)$ holds in the extended market, thus by Theorem \ref{FTAP3}, there exists an equivalent collective martingale measure $Q=(Q^1,\dots, Q^N)$ for the extended markets, with bounded densities w.r.t. $P'$, and in particular belonging to $\mathcal{M}^\varphi_e(\mcY)$ as a consequence. Then,  $X^{J+i}_t = E_{Q^i}[X^{J+i}_T | \mathcal{F}^i_t]=E_{Q^i}[ f^i | \mathcal{F}^i_t]$. By definition
\[
\Pi^i_ f = X^{J+i}_0 = E_{Q^i}[X^{J+i}_T | \mathcal{F}_0] = E_{Q^i}[   f^i].
\]
Conversely, take $(E_{Q^1}[   f^1],\dots, E_{Q^N}[   f^N])$ for $ Q \in \mathcal{M}^\varphi_e(\mcY)$. Define for $i=1,\dots,N$ the $Q^i-$ martingales $X^{J+i}_t = {E}_{Q^i}[   f^i | \mathcal{F}^i_t]$. Since $Q \in \mathcal{M}^\mcY_e$, and $Q^i$ is an equivalent martingale measure for the extended market $(X^j)_{j\in(i)\cup\{J+i\}}$, $Q$ acts as an equivalent collective martingale measure for the extended markets, which thus does not admit $\mcY$-collective arbitrages, by Theorem \ref{Th1}.
Finally, $X^{J+i}_T = f^i$ for mere measurability and integrability arguments, thus
thus $\Pi_ f \in \Pi( f)$ by definition.
\end{proof}

Suppose that $f=(f^1,\dots,f^N) \in L^{0 }(\Omega, \mathbf{F}_T,P)$ is not collectively replicable. Differently from the classical case, see for example \cite{FollmerSchied2} Theorem 5.32, the set $\Pi(f)$ is not necessarily open (see the Example \ref{notopen}). Moreover $\Pi(f)$ can not be closed, at least under weak conditions on $f$ (see Proposition \ref{notclosed}).

\begin{example}\label{notopen}

Consider the example 7.2 \cite{BDFFM25} with $|\Omega|=6$, two periods and $N=2$ agents, each investing in one single stock. Take as contingent claims $g^1=0$ and $g^2=(0,0,0,0,1,1)$, $g=(g^1,g^2)$. Using the set of collective martingale measures on page 49 \cite{BDFFM25},  we compute $E_{Q^1}[g^1]+E_{Q^2}[g^2]=\frac{1}{2}q,$ with  $0< q< 1$. Then $\rho^{\mcY}_-(g) =0$ and $ \rho^{\mcY}_+(g)=\frac{1}{2}$ and $g$ is not $\mcY$-collectively replicable. Moreover, using Theorem \ref{BaseFTAPII} $$\Pi(g)=\{(0,\frac{1}{2}q) \mid 0<q<1 \} \subset \mathbb R^2.$$
Thus this examples shows that for a non collectively replicable $g$, $\Pi(g)$ is not necessarily an open set in $\mathbb R^2 $. Clearly, by Theorem \ref{BaseFTAPII}, for any claim $g$ we have $$p \in \Pi(g) \Rightarrow \rho^{\mcY}_-(g)\leq \sum_{i=1}^Np^i \leq \rho^{\mcY}_+(g)$$ 
At the same, $\Pi(g)$ is not closed in general, as shown in Proposition \ref{notclosed}

\end{example}

\begin{proposition}
\label{prepcftap2}
 Let Assumption \ref{ass2} hold. Then $f \in L^{0 }(\Omega, \mathbf{F}_T , P)$ is $\mcY$-collectively replicable if and only if $\rho^{ \mcY}_-( f)=\rho^{ \mcY}_+( f)$.
\end{proposition}
\begin{proof}
Recall from Lemma \ref{lemmarho} Item 6 that by $\mathbf{NCA}(\mcY)$ we have  $\rho^{ \mcY}_-( f)\leq\rho^{ \mcY}_+( f)$.
We start by proving that if $f$ is collectively replicable then $\rho^{ \mcY}_-( f)=\rho^{ \mcY}_+( f)$. 
Let  $\varphi\in L^{0}(\Omega, \mathbf{F}_T,P)$ be any given random vector  satisfying 
$\max \{|f^i|,|Y^i_1|,\ldots, |Y^i_R|\}\leq \varphi^i$ for every $i$. Take $Q\in \mathcal{M}^\varphi_e(\mcY)$, which is nonempty by Theorem \ref{duality:R}. Write $f=x+k+Y$ with obvious notation from replicability. Since $f,Y\in L^{1}(\Omega, \mathbf{F}_T,\mathbf{Q})$, by definition of $\mathcal{M}^\varphi_e(\mcY)$ we see that  $(k^i)^-\in L^1(Q^i)$ for every $i$. \cite{FollmerSchied2} Theorem 5.14 implies that $E_{Q^i}[k^i]=0$ for every $i$. As a consequence, using also $Q\in \mathcal{M}^\varphi_e(\mcY)$, we can write
    $$\sum_{i=1}^NE_{Q^i}[f^i]=\sum_{i=1}^N\Big(x^i+E_{Q^i}[k^i]
    +E_{Q^i}[Y^i]\Big)=\sum_{i=1}^Nx^i.$$
    Hence $Q\mapsto \sum_{i=1}^NE_{Q^i}[f^i]$ is constant on $\mathcal{M}^\varphi_e(\mcY)$, so using \eqref{pricing:hedging}
    $$
  \rho^{\mcY}_+(f)=\sup \left \{\sum_{i=1}^N E_{Q^i}[f^i]\mid Q\in \mathcal{M}^{\varphi}_e(\mcY) \right \}=\inf\left \{\sum_{i=1}^N E_{Q^i}[f^i]\mid Q\in \mathcal{M}^{\varphi}_e(\mcY) \right \}=\rho^{\mcY}_-(f).$$
We now prove that for a non $\mcY$-collectively replicable $f$ we have
$\rho^{ \mcY}_-( f) < \rho^{ \mcY}_+( f)$.
 By the definitions, $\rho^{ \mcY}_-( f) \leq \rho^{ \mcY}_+( f)$. Assume by contradiction $\rho^{ \mcY}_-( f) = \rho^{ \mcY}_+( f)$. By Lemma \ref{lemmarho} Item 5 both these then need to be finite and  attained. Recall from  \eqref{piY} the definition of $\pi^{\mcY}_+$.  From \cite{BDFFM25} Proposition 4.9  $\rho^{\mcY}_+( f)=N\cdot \pi^{\mcY}_+( f)$ and $\pi^{\mcY}_+( f)$ is attained iff $\rho^{\mcY}_+(f)$ is attained.  Thus we also have $\pi^{ \mcY}_-( f) = \pi^{ \mcY}_+( f):=x \in \mathbb R$ and  there exists $k_\pm\in {\sf X}_{i=1} ^{N} K_i$, $Y_\pm\in\mcY$ such that 
  \[x+ k_-^i - Y_-^i \leq  f^i \leq x+ k_+^i + Y_+^i\quad \forall \,i=1,\ldots,N.\]
  As a consequence 
  \[0\leq  k_+^i-k_-^i + Y_+^i + Y_-^i \quad \forall \,i=1,\ldots,N,\]
  which by $\mathbf{NCA(\mcY)}$ implies
  \[0 =  k_+^i-k_-^i + Y_+^i+Y_-^i \quad \forall \,i=1,\ldots,N,\]
  i.e. 
  \[k_-^i - Y_-^i =  f^i-x = k_+^i + Y_+^i\quad \forall \,i=1,\ldots,N,\]
  which is not possible as $ f$ is not $\mcY$-collectively replicable.  
\end{proof}

\begin{proposition}
\label{prop:attainrepl}
Let Assumption \ref{ass2} hold. 
Take  $f\in L^{0}(\Omega, \mathbf{F}_T,P)$. Let  $\varphi\in L^{0}(\Omega, \mathbf{F}_T,P)$ be any given random vector  satisfying 
$\max \{|f^i|,|Y^i_1|,\ldots, |Y^i_R|\}\leq \varphi^i$ for every $i$.
    Suppose that both the primal problem \eqref{super:rho} and the dual one \eqref{pricing:hedging} are attained. Then $f$ is $\mcY$-collectively replicable.
    The same implication holds for $\rho^{\mcY}_-$ as well.
\end{proposition}
\begin{proof}
    Suppose $\mcR( f)=\sum_{i=1}^NE_{Q^i}[   f^i]$  for some $ Q \in \mathcal{M}^\varphi_e(\mcY)$. Since $\mcR( f)$ is finite, it is  attained (Lemma \ref{lemmarho} Item 5) and thus for some $x\in \R^N$ with $\sum_i x^i=\mcR( f)$, $k\in {\sf X}_{i=1} ^{N} K_i$, $Y\in\mcY$ we have $x^i+k^i+Y^i\geq  f^i$ for every $i$. By setting $\widehat{Y}^i=x^i-E_{Q^i}[   f^i]+Y^i$ we also have that $E_{Q^i}[   f^i]+k^i+\widehat{Y}^i\geq  f^i$ for every $i$, with $\widehat{Y}\in\mcY$. This immediately gives $E_{Q^i}[  \widehat{Y}^i]\geq 0$ for every $i$. At the same time if the inequality $E_{Q^i}[   f^i]+k^i+\widehat{Y}^i\geq  f^i$ were strict with positive ($P$ or $Q^i$ equivalently) probability for some index $h$, then $E_{Q^h}[  \widehat{Y}^h]>0$ with a similar computation. This would yield $\sum_jE_{Q^j}[  \widehat{Y}^j]>0$, a contradiction with $Q\in\mathcal{M}^\varphi_e(\mcY) $. We then infer $E_{Q^i}[   f^i]+k^i+\widehat{Y}^i=  f^i$, for every $i$, $Q^i$ (hence $P$ by equivalence) a.s..
\end{proof}

\begin{proposition}\label{notclosed}
Let Assumption \ref{ass2} hold. 
   If $f\in L^0(\Omega, \mathbf{F}_T, P)$ is not $\mcY$-collectively replicable and if $\pii( f^i)<+\infty$ and $\rho_{i,+}(-f^i)<+\infty$ for every $i$ then $\Pi(f)$ is not closed. 
\end{proposition}
\begin{proof}
Every $p \in \Pi(f) $ satisfies
$$-\rho_{i,+}(-f^i)\leq \inf \left \{E_{Q^i}[f^i]\mid Q\in \mathcal{M}^{\varphi}_e(\mcY) \right \}\leq p^i \leq \sup \left \{E_{Q^i}[f^i]\mid Q\in \mathcal{M}^{\varphi}_e(\mcY) \right \}\leq \pii( f^i) $$
and thus $\Pi(f)$ is bounded in $\mathbb R^N$. If $\Pi(f)$ were closed then it would be compact and by Theorem \ref{BaseFTAPII} together with \eqref{pricing:hedging} we would have  $
  \rho^{\mcY}_+(f)=\sum_{i=1}^N E_{Q^i}[f^i]$, for some 
 $Q\in M^\varphi_e(\mcY)$, for $\max \{|f^i|,|Y^i_1|,\ldots, |Y^i_R|\}\leq \varphi^i$. This in turns would imply by Proposition \ref{prop:attainrepl} that $f$ is $\mcY$-collectively replicable, a contradiction.
\end{proof}

Theorem \ref{completeTH} is a consequence of the following proposition. We write $e_1,\dots, e_N$ for the canonical basis of $\R^N$
\begin{proposition}\label{propclosed}
    Let Assumption \ref{ass2} hold. Then the following are equivalent:
    \begin{enumerate}
        \item Every $ f\in L^{0}(\Omega, \mathbf{F}_T,P)$ is $\mcY$-collectively replicable;
        \item Every $ f\in L^{\infty}(\Omega, \mathbf{F}_T,P)$ is $\mcY$-collectively replicable;
        \item  $1_Ae_j$ is $\mcY$-collectively replicable for every $j=1,\dots, N$ and every $A\in \mcF_T^j$;
        \item $\Me$ is a singleton;
        \item  $\rho^{ \mcY}_\pm$ coincide on $L^{0}(\Omega, \mathbf{F}_T,P)$.
    \end{enumerate}
\end{proposition}
\begin{proof}
    The implications $1\Rightarrow 2\Rightarrow 3$ are obvious. We show $3\Rightarrow 4$. Take $Q,\widehat{Q}\in \Me$. From Proposition \ref{prepcftap2} we have $\rho^{ \mcY}_-( 1_Ae_j)=\rho^{ \mcY}_+( 1_Ae_j)$. A the same time, choosing $\varphi^i=\max \{|Y^i_1|,\ldots, |Y^i_R|\}$, Remark \ref{remdegenerate}, together with Theorem \ref{duality:R} allows us to write
    $$\rho^{ \mcY}_-( 1_Ae_j)\leq E_{Q^j}[1_A]=\sum_{i=1}^NE_{Q^i}[1_Ae_j]\leq \rho^{ \mcY}_+( 1_Ae_j)$$ and the same holds for $E_{\widehat{Q}^j}[1_A]$, so that $E_{Q^j}[1_A]=E_{\widehat{Q}^j}[1_A]$. Since this can be repeated for every $j$ and $A\in\mcF_T^j$, it follows that $\Me$ (which is nonempty) also consists of at most one element.
    
    We come to $4\Rightarrow 5$. Pick $f\in L^{0}(\Omega, \mathbf{F}_T,P)$ and $\varphi\in L^{0}(\Omega, \mathbf{F}_T,P)$ with
$\max \{|f^i|,|Y^i_1|,\ldots, |Y^i_R|\}\leq \varphi^i$. As argued before, $ \mathcal{M}^\varphi_e(\mcY)\neq \emptyset$ by $\mathbf{NCA}(\mcY)$ and at the same time $ \mathcal{M}^\varphi_e(\mcY)\subseteq \Me$, the latter being a singleton. Thus, $\mathcal{M}^\varphi_e(\mcY)$ is itself a singleton, say $\{\widehat{Q}\}$. Since $\rho^{\mcY}_-(f)=-\rho^{\mcY}_+(-f)$, \eqref{pricing:hedging} implies $$\rho^{\mcY}_+(f)=\sum_{i=1}^N E_{\widehat{Q}^i}[f^i]=\rho^{\mcY}_-(f).$$
Finally, implication $5\Rightarrow 1$ follows immediately from Proposition \ref{prepcftap2}.
\end{proof}

\section{Examples}
\label{sec:examples}
We describe a simple market with two agents and two stocks, and where each agent $i$ may invest only in the stock $X^i$ and in the riskless asset.  In this example there exists a global arbitrage ($M_e := M_e^1\cap  M_e^2 =\emptyset$); each single market is arbitrage free and incomplete ($| M_e^1|=\infty$ and $|M_e^2|=\infty$); $\mathbf {NCA}(\mcY) $ holds true and the market is collectively complete ($|\mathcal{M}_e(\mcY)|=1$), so that there exist a unique price vector for each vector of contingent claims. 
Consider the setup described in Figure \ref{figtree}, with common filtration $$\mcF_0=\{\emptyset, \Omega\},\, \mcF_1=\sigma(A_1=\{\omega_1,\omega_2\},A_2=\{\omega_3,\omega_4\},A_3=\{\omega_5,\omega_6\}),\,\mcF_2=\mathcal{P}(\Omega).$$
We show that every vector of contingent claims is collectively replicable for $\mcY=\mcY_0\cap L^0(\Omega,\mcF_1,P)$.
To this end we first focus on the time step $\{0,1\}$. Take any $g_1,g_2$ which are $\mcF_1$ measurable. Let  $g^i=(g^i_j)_{j=1,2,3}, \,g^i_j\in\R$ on $A_j,j=1,2,3$. Similarly, $Y=(Y^1,Y^2)\in\mcY$ is such that $Y^2=-Y^1$ and $Y^1=y^1_j$ on $A_j,j=1,2,3$. We show that there exist $(B^1,B^2,h,H,y^1_1,y^1_2,y^1_3)\in\R^6$ such that
$$B^1+hX^1(\omega)+Y^1(\omega)=g^1(\omega),\quad B^2+HX^2(\omega)-Y^1(\omega)=g^2(\omega)\quad \forall\omega\in\Omega.$$
This reads:
$$\begin{cases}
    B^1+4h+0H+y^1_1+0y^1_2+0y^1_3&=g^1_1\\
    B^1+4h+0H+0y^1_1+y^1_2+0y^1_3&=g^1_2\\
    B^1+12h+0H+0y^1_1+0y^1_2+y^1_3&=g^1_3\\
    B^2+0h+16H-y^1_1+0y^1_2+0y^1_3&=g^2_1\\
    B^2+0h+8H+0y^1_1-y^1_2+0y^1_3&=g^2_2\\
    B^2+0h+8H+0y^1_1+0y^1_2-y^1_3&=g^2_3\\    
\end{cases}\text{ or } \begin{pmatrix}
    1 & 0 & 4 & 0 & 1 & 0 & 0 \\
    1 & 0 & 4 & 0 & 0 & 1 & 0 \\
    1 & 0 & 12 & 0 & 0 & 0 & 1 \\
    0 & 1 & 0 & 16 & -1 & 0 & 0 \\
    0 & 1 & 0 & 8 & 0 & -1 & 0 \\
    0 & 1 & 0 & 8 & 0 & 0 & -1 \\
\end{pmatrix}
\begin{pmatrix}
    B^1 \\ B^2 \\ h \\ H \\ y^1_1 \\ y^1_2 \\ y^1_3
\end{pmatrix}
=
\begin{pmatrix}
    g^1_1 \\ g^1_2 \\ g^1_3 \\ g^2_1 \\ g^2_2 \\ g^2_3
\end{pmatrix} $$
The rank of the matrix of the system is $6$, a solution exists. In particular, one such a solution can be always obtained even asking $B^1=B^2$, since 
$$\mathrm{det}\begin{pmatrix}
    1 & 4 & 0 & 1 & 0 & 0 \\
    1 & 4 & 0 & 0 & 1 & 0 \\
    1 & 12 & 0 & 0 & 0 & 1 \\
    1 & 0 & 16 & -1 & 0 & 0 \\
    1 & 0 & 8 & 0 & -1 & 0 \\
    1 & 0 & 8 & 0 & 0 & -1 \\
\end{pmatrix}=-128.$$
Now observe that in the time step $\{1,2\}$ each of the two markets is (separately) complete. The arguments above allow us to conclude that each $\mcF_2$-measurable pair $(g_1,g_2)$ is collectively replicable. One can easily compute the relevant sets of equivalent martingale measures:
\begin{align*}
    M_e^1&=\bigg\{\bigg(\frac12 \Big(\frac12-p\Big), \frac12\Big(\frac12-p\Big), \frac12 p, \frac12 p, \frac16, \frac13\bigg),\, p\in\bigg(0,\frac{1}{2}\bigg)\bigg\},\\
     M_e^2&=\bigg\{\bigg(\frac18, \frac18, \frac12\Big(\frac34-q\Big), \frac12\Big(\frac34-q\Big), \frac23q, \frac13q\bigg),\, q\in\bigg(0,\frac{3}{4}\bigg)\bigg\},\\
     \mathcal{M}_e(\mcY)&=\{(Q^1,Q^2)\} \not=\emptyset, \quad M_e := M_e^1\cap  M_e^2 =\emptyset,
\end{align*}
where $Q^1\in M_e^1$ corresponds to $p=\frac14$ and $Q^2\in M^2_e$ to the choice $q=\frac12$. This is clearly consistent with CFTAP II. Note that although there exists a global arbitrage ($M_e = M_e^1 \cap  M_e^2 =\emptyset$) and $Q^1\neq Q^2$ on $\mcF_2$, we do have $Q^1=Q^2$ on $\mcF_1$, which is equivalent to the polarity condition for the exchanges $\mcY$.
\begin{figure}
\begin{center}
\tikzstyle{level 1}=[level distance=2cm, sibling distance=2.5cm,->]
\tikzstyle{level 2}=[level distance=1.5cm, sibling distance=1cm,->]

\tikzstyle{bag} = [text width=1.5em, text centered]
\tikzstyle{end} = []

\begin{tikzpicture}[grow=right, sloped]
\node[bag](c1){$8$}
    child {
        node[bag]{$12$}        
            child {
                node[end, label=right:
                    {$9$}](y16) {}
                edge from parent
                node[above] {}
                node[below]  {$\blue{2/3}$}
            }
            child {
                node[end, label=right:
                    {$18$}](y15) {}
                edge from parent
                node[above] {$\blue{1/3}$}
                node[below]  {}
            }
            edge from parent 
            node[above] {}
            node[below]  {$\blue{(1/2)}$}
    }
    child {
        node[bag]{$4$}        
            child {
                node[end, label=right:
                    {$3$}](y14) {}
                edge from parent
                node[above] {}
                node[below]  {$\blue{1/2}$}
            }
            child {
                node[end, label=right:
                    {$5$}](y13) {}
                edge from parent
                node[above] {$\blue{1/2}$}
                node[below]  {}
            }
            edge from parent 
            node[above] {$\blue{p}$}
            node[below]  {}
    }
    child {
        node[bag] {$4$}        
        child {
                node[end, label=right:
                    {$2$}] (y12){}
                edge from parent
                node[above] {}
                node[below]  {$\blue{1/2}$}
            }
            child {
                node[end, label=right:
                    {$6$}](y11) {}
                edge from parent
                node[above] {$\blue{1/2}$}
                node[below]  {}
            }
        edge from parent         
            node[above] {$\blue{1/2-p}$}
            node[below]  {}
    };

\node[bag](y21) at ([xshift=1.1cm]y11) {$\red{\omega_1}$};
\node[bag](y22) at ([xshift=1.1cm]y12) {$\red{\omega_2}$};
\node[bag](y23) at ([xshift=1.1cm]y13) {$\red{\omega_3}$};
\node[bag](y24) at ([xshift=1.1cm]y14) {$\red{\omega_4}$};
\node[bag](y25) at ([xshift=1.1cm]y15) {$\red{\omega_5}$};
\node[bag](y26) at ([xshift=1.1cm]y16) {$\red{\omega_6}$};
% %%%%%Secondo:
\node[bag](c2) at ([xshift=6cm]c1){$10$}
    child {
        node[bag]{$8$}        
            child {
                node[end, label=right:
                    {$12$}](z16) {}
                edge from parent
                node[above] {}
                node[below]  {$\blue{1/3}$}
            }
            child {
                node[end, label=right:
                    {$6$}](z15) {}
                edge from parent
                node[above] {$\blue{2/3}$}
                node[below]  {}
            }
            edge from parent 
            node[above] {}
            node[below]  {$\blue{q}$}
    }
    child {
        node[bag]{$8$}        
            child {
                node[end, label=right:
                    {$4$}](z14) {}
                edge from parent
                node[above] {}
                node[below]  {$\blue{1/2}$}
            }
            child {
                node[end, label=right:
                    {$12$}](z13) {}
                edge from parent
                node[above] {$\blue{1/2}$}
                node[below]  {}
            }
            edge from parent 
            node[above] {$\blue{3/4-q}$}
            node[below]  {}
    }
    child {
        node[bag] {$16$}        
        child {
                node[end, label=right:
                    {$12$}] (z12){}
                edge from parent
                node[above] {}
                node[below]  {$\blue{1/2}$}
            }
            child {
                node[end, label=right:
                    {$20$}](z11) {}
                edge from parent
                node[above] {$\blue{1/2}$}
                node[below]  {}
            }
        edge from parent         
            node[above] {$\blue{1/4}$}
            node[below]  {}
    };

\node[bag](z21) at ([xshift=1.1cm]z11) {$\red{\omega_1}$};
\node[bag](z22) at ([xshift=1.1cm]z12) {$\red{\omega_2}$};
\node[bag](z23) at ([xshift=1.1cm]z13) {$\red{\omega_3}$};
\node[bag](z24) at ([xshift=1.1cm]z14) {$\red{\omega_4}$};
\node[bag](z25) at ([xshift=1.1cm]z15) {$\red{\omega_5}$};
\node[bag](z26) at ([xshift=1.1cm]z16) {$\red{\omega_6}$};

\end{tikzpicture}
\end{center}
\caption{Tree for the stocks $(X^1,X^2)$ at times $t=0,1,2$.}
\label{figtree}
\end{figure}
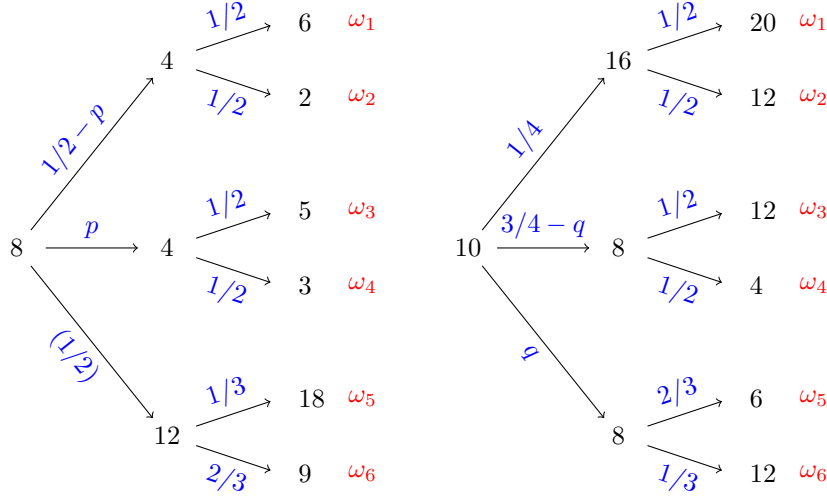

\subsection{Example: time consistency may fail}
\label{sec:timeconsist}
Again we consider a simple market with two agents and two stocks, and where each agent $i$ may invest only in the stock $X^i$ and in the riskless asset. 
The dynamics of $X^1,X^2$ are depicted in Fig. \ref{fig:inc}. 
\\ We observe that there exists a global arbitrage opportunity in $(0,1)$  as $(\Delta X^1_1+\Delta X^2_1)(\omega) \geq 0$ for any $\omega\in\{\omega_1,\ldots, \omega_6\}$ and the inequality is strict on $\omega_1,\omega_2$. On the other hand there does not exist any global arbitrage opportunity in the interval  $(1,2)$.
\\ By selecting $\mathcal{Y}(1)=\R^N_0$ it is easy to verify that $\mathbf{NCA}(\mathcal{Y}(1))$ holds in the interval $(0,1)$.

We show the existence of a collective arbitrage in the interval $(0,2)$. Suppose that agent $1$ and agent $2$ buy one unit of stock i.e. $H_1^2=H_1^1=1$ and then sell it, i.e. $H_2^2=H_2^1=0$. Choose $\mathcal{Y}(2)=\{ Y \in L^0(\Omega, \mathbf{F}_1, P) \mid \sum_{i=1}^N Y^i = 0  \}$ and let $Y^1(\omega_1,\omega_2)=0,\, Y^1(\omega_3,\omega_4)=+1,\, Y^1(\omega_5, \omega_6)=-1$ and let  $Y^2=-Y^1$. 
Then for any $\omega\in \{\omega_1,\ldots, \omega_6\}$ the value of this strategy is

\begin{eqnarray*}
    H^i_1\Delta X^i(\omega)+ Y^i(\omega)\geq 0 & \text{ and } & H_1^i\Delta X^i(\omega_{1,2})+ Y^i(\omega_{1,2})> 0 \quad i=1,2.
\end{eqnarray*}

The above conclusions could be also derived from the dual point of view. Indeed, 
let $M^i_e(s,t)$ be the set of equivalent martingale measures for asset $i$ in the interval $(s,t)$ and denote with $M^{\mcY}_e (s,t)$ the set of collective martingale measures for $\mcY$ in the interval $(s,t)$. Denote with $\mathbf{A}(s,t)$ (resp. $\mathbf{NA}(s,t)$) the existence (resp. the absence) of a global arbitrage in the interval $(s,t)$, and with $\mathbf{CA}^{\mcY}(s,t)$) (resp. $\mathbf{NCA}^{\mcY}(s,t)$) the existence (resp. the absence) of a collective arbitrage with respect to $\mcY$  in the interval $(s,t)$. Then 
\begin{enumerate}
    \item $M^1_e(0,1)\cap M^2_e(0,1)= \emptyset  \Rightarrow A(0,1)$; 
    \item $M^1_e(0,2)\cap M^2_e(0,2)= \emptyset  \Rightarrow A(0,2)$;
    \item $M^1_e(1,2)=\{1/2, 1/2,1/2,1/2,1/2,1/2\}=M^2_e(1,2) \Rightarrow M^1_e(1,2)\cap M^2_e(1,2)\not= \emptyset \Rightarrow \mathbf{NA}(1,2)$ and this implies $\mathbf{NCA}^{\mcY(2)}(1,2)$ (see \eqref{Implications});
    \item for $\mathcal{Y}(1)=\R^N_0$, $M^{\mcY(1)}_e (0,1) =M^1_e(0,1)\times M^2_e(0,1) \not=\emptyset \Rightarrow \mathbf{NCA}^{\mcY(1)}(0,1)$;
    \item for $\mathcal{Y}(1)=\R^N_0$ and $\mathcal{Y}(2)=\{ Y \in L^0(\Omega, \mathbf{F}_1, P) \mid \sum_{i=1}^N Y^i = 0  \}$, consider   $\mcY_{1:2} $ as defined in Section \ref{sec:selfin}. Then  $M^{\mcY_{1:2}}_e (0,2)=\left \{(Q^1,Q^2) \mid Q^1=Q^2 \text{ on } \mathcal F_1 \right \} =\emptyset \Rightarrow \mathbf{CA}^{\mcY_{1:2}}(0,2)$.
\end{enumerate}   

The conditions $\mathbf{A}(0,1)$, $\mathbf{NA}(1,2)$ and $\mathbf{A}(0,2)$  are consistent with the well-known time consistency property of classical arbitrage, which asserts that an arbitrage opportunity exists in the multi-period market if and only if an arbitrage opportunity exists in some single-period sub-market. This property was also shown to hold for collective (multi-period) arbitrage when exchanges are restricted to a single time $t \in \{0,\dots,T\}$ (see Section 6.3 of \cite{BDFFM25}). However, the present example demonstrates that this time consistency property fails to hold for collective arbitrage when exchanges are permitted at multiple times. Specifically, we have $\mathbf{NCA}^{\mcY(1)}(0,1)$, $\mathbf{NCA}^{\mcY(2)}(1,2)$, $\mathbf{CA}^{\mcY_{1:2}}(0,2)$.

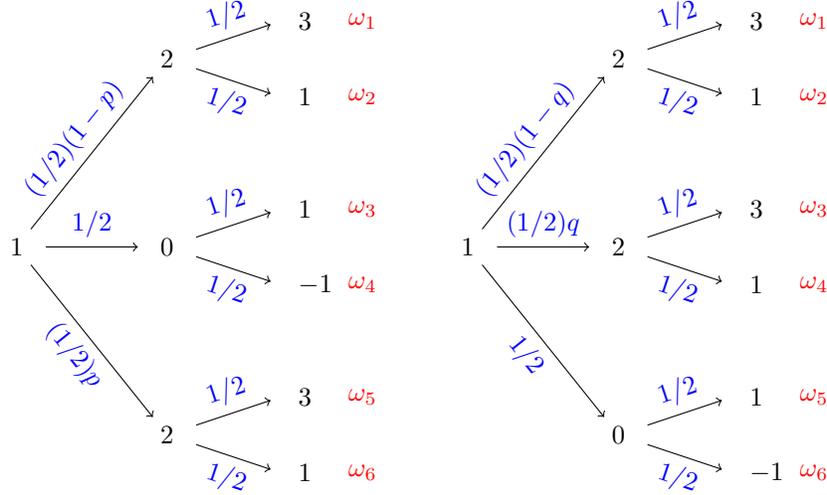
\begin{figure}
\begin{center}
\tikzstyle{level 1}=[level distance=2cm, sibling distance=2.5cm,->]
\tikzstyle{level 2}=[level distance=1.5cm, sibling distance=1cm,->]

\tikzstyle{bag} = [text width=1.5em, text centered]
\tikzstyle{end} = []

\begin{tikzpicture}[grow=right, sloped]
\node[bag](c1){$1$}
    child {
        node[bag]{$2$}        
            child {
                node[end, label=right:
                    {$1$}](y16) {}
                edge from parent
                node[above] {}
                node[below]  {$\blue{1/2}$}
            }
            child {
                node[end, label=right:
                    {$3$}](y15) {}
                edge from parent
                node[above] {$\blue{1/2}$}
                node[below]  {}
            }
            edge from parent 
            node[above] {}
            node[below]  {$\blue{(1/2)p}$}
    }
    child {
        node[bag]{$0$}        
            child {
                node[end, label=right:
                    {$-1$}](y14) {}
                edge from parent
                node[above] {}
                node[below]  {$\blue{1/2}$}
            }
            child {
                node[end, label=right:
                    {$1$}](y13) {}
                edge from parent
                node[above] {$\blue{1/2}$}
                node[below]  {}
            }
            edge from parent 
            node[above] {$\blue{1/2}$}
            node[below]  {}
    }
    child {
        node[bag] {$2$}        
        child {
                node[end, label=right:
                    {$1$}] (y12){}
                edge from parent
                node[above] {}
                node[below]  {$\blue{1/2}$}
            }
            child {
                node[end, label=right:
                    {$3$}](y11) {}
                edge from parent
                node[above] {$\blue{1/2}$}
                node[below]  {}
            }
        edge from parent         
            node[above] {$\blue{(1/2)(1-p)}$}
            node[below]  {}
    };

\node[bag](y21) at ([xshift=1.1cm]y11) {$\red{\omega_1}$};
\node[bag](y22) at ([xshift=1.1cm]y12) {$\red{\omega_2}$};
\node[bag](y23) at ([xshift=1.1cm]y13) {$\red{\omega_3}$};
\node[bag](y24) at ([xshift=1.1cm]y14) {$\red{\omega_4}$};
\node[bag](y25) at ([xshift=1.1cm]y15) {$\red{\omega_5}$};
\node[bag](y26) at ([xshift=1.1cm]y16) {$\red{\omega_6}$};
% %%%%%Secondo:
\node[bag](c2) at ([xshift=6cm]c1){$1$}
    child {
        node[bag]{$0$}        
            child {
                node[end, label=right:
                    {$-1$}](z16) {}
                edge from parent
                node[above] {}
                node[below]  {$\blue{1/2}$}
            }
            child {
                node[end, label=right:
                    {$1$}](z15) {}
                edge from parent
                node[above] {$\blue{1/2}$}
                node[below]  {}
            }
            edge from parent 
            node[above] {}
            node[below]  {$\blue{1/2}$}
    }
    child {
        node[bag]{$2$}        
            child {
                node[end, label=right:
                    {$1$}](z14) {}
                edge from parent
                node[above] {}
                node[below]  {$\blue{1/2}$}
            }
            child {
                node[end, label=right:
                    {$3$}](z13) {}
                edge from parent
                node[above] {$\blue{1/2}$}
                node[below]  {}
            }
            edge from parent 
            node[above] {$\blue{(1/2)q}$}
            node[below]  {}
    }
    child {
        node[bag] {$2$}        
        child {
                node[end, label=right:
                    {$1$}] (z12){}
                edge from parent
                node[above] {}
                node[below]  {$\blue{1/2}$}
            }
            child {
                node[end, label=right:
                    {$3$}](z11) {}
                edge from parent
                node[above] {$\blue{1/2}$}
                node[below]  {}
            }
        edge from parent         
            node[above] {$\blue{(1/2)(1-q)}$}
            node[below]  {}
    };

\node[bag](z21) at ([xshift=1.1cm]z11) {$\red{\omega_1}$};
\node[bag](z22) at ([xshift=1.1cm]z12) {$\red{\omega_2}$};
\node[bag](z23) at ([xshift=1.1cm]z13) {$\red{\omega_3}$};
\node[bag](z24) at ([xshift=1.1cm]z14) {$\red{\omega_4}$};
\node[bag](z25) at ([xshift=1.1cm]z15) {$\red{\omega_5}$};
\node[bag](z26) at ([xshift=1.1cm]z16) {$\red{\omega_6}$};

\end{tikzpicture}
\end{center}
\caption{Tree for the stocks $(X^1,X^2)$ at times $t=0,1,2$.}
\label{fig:inc}
\end{figure}

\paragraph{Acknowledgements.}

The authors are funded by the European Union - NextGenerationEU through the Italian Ministry of University and Research under the National Recovery and Resilience Plan (PNRR) - Mission 4 Education and research - Component 2 From research to business - Investment 1.1 Notice Prin 2022 - DD N. 104 del 2/2/2022, title Entropy Martingale Optimal Transport and McKean-Vlasov equations - codice progetto 2022K28KB7 - CUP G53D23001830006. A. Doldi and M. Frittelli are members of GNAMPA-INDAM.\\All the authors declare no conflicts of interest.

\bibliographystyle{abbrv}  
\bibliography{BibAll}

\end{document}